\documentclass[a4paper,11pt,oneside]{article}
\usepackage{latexsym}
\usepackage[english]{babel}
\usepackage{fancyhdr}
\usepackage[left=2.5cm, right=2.5cm]{geometry}

\usepackage[mathscr]{eucal}
\usepackage{amsmath}
\usepackage{mathrsfs}
\usepackage{amsthm}
\usepackage{amsfonts} % included in \usepackage{iopams}
\usepackage{amssymb} % \usepackage{iopams}
\usepackage{amscd}
\usepackage{bbm}
\usepackage{graphicx}
\usepackage{graphics}

\usepackage{bm}

\newcommand{\ii}{\mathrm{i}}
\newcommand{\cH}{\mathcal{H}}

\newcommand{\ud}{\mathrm{d}}
 \newtheorem{thm}{Theorem}[section]
 \newtheorem{cor}[thm]{Corollary}
 
 \newtheorem{prop}[thm]{Proposition}
 \theoremstyle{definition}
 
 \theoremstyle{remark}

 \numberwithin{equation}{section}

\begin{document}\hyphenation{Cou-lomb}

%-------------------------------------------------------------------------
% editorial commands: to be inserted by the editorial office
%
%\firstpage{1} \volume{228} \Copyrightyear{2004} \DOI{003-0001}
%
%
%\seriesextra{Just an add-on}
%\seriesextraline{This is the Concrete Title of this Book\br H.E. R and S.T.C. W, Eds.}
%
% for journals:
%
%\firstpage{1}
%\issuenumber{1}
%\Volumeandyear{1 (2004)}
%\Copyrightyear{2004}
%\DOI{003-xxxx-y}
%\Signet
%\commby{inhouse}
%\submitted{March 14, 2003}
%\received{March 16, 2000}
%\revised{June 1, 2000}
%
\title{Discrete spectra for critical Dirac-Coulomb Hamiltonians}% Force line breaks with \\
%\accepted{July 22, 2000}
%
%
%
%---------------------------------------------------------------------------
%Insert here the title, affiliations and abstract:
%

%----------Author 1
\author{Matteo Gallone\footnote{International School for Advanced Studies -- SISSA, via Bonomea 265, 34136 Trieste, Italy. e-mail 
\texttt{mgallone@sissa.it}} $\,$ and Alessandro Michelangeli\footnote{International School for Advanced Studies -- SISSA, via Bonomea 265 34136 Trieste, Italy.
e-mail 	\texttt{alemiche@sissa.it}}}

%----------classification, keywords, date

% 
% 47B25
% Symmetric and selfadjoint operators (unbounded)
% 
% 47N20
% Applications to differential and integral equations
% 
% 47N50
% Applications in the physical sciences
% 
% 81Q10
% Selfadjoint operator theory in quantum theory, including spectral
% analysis

\date{\today}
%----------additions
%\dedicatory{To my boss}
%%% ----------------------------------------------------------------------
\maketitle

\begin{abstract}
The one-particle Dirac Hamiltonian with Coulomb interaction is known to be realised, in a regime of large (critical) couplings, by an infinite multiplicity of distinct self-adjoint operators, including a distinguished, physically most natural one. For the latter, Sommerfeld's celebrated fine structure formula provides the well-known expression for the eigenvalues in the gap of the continuum spectrum. Exploiting our recent general classification of all other self-adjoint realisations, we generalise Sommerfeld's formula so as to determine the discrete spectrum of all other self-adjoint versions of the Dirac-Coulomb Hamiltonian. Such discrete spectra display naturally a fibred structure, whose bundle covers the whole gap of the continuum spectrum.
% An article usually includes an abstract, a concise summary of the work
% covered at length in the main body of the article. It is used for
% secondary publications and for information retrieval purposes. 
% %
% Valid PACS numbers may be entered using the \verb+\pacs{#1}+ command.


\end{abstract}

% PACS, the Physics and Astronomy
                             % Classification Scheme.

% 02.30.Gp 	Special functions
% 02.30.Hq 	Ordinary differential equations
% 02.30.Sa 	Functional analysis
% 02.30.Tb 	Operator theory
% 03.65.Pm        Dirac equation
% 32.30.-r 	Atomic spectra 

                              %display desired
\maketitle

% \begin{quotation}
% The ``lead paragraph'' is encapsulated with the \LaTeX\ 
% \verb+quotation+ environment and is formatted as a single paragraph before the first section heading. 
% (The \verb+quotation+ environment reverts to its usual meaning after the first sectioning command.) 
% Note that numbered references are allowed in the lead paragraph.
% %
% The lead paragraph will only be found in an article being prepared for the journal \textit{Chaos}.
% \end{quotation}

\section{\label{sec:Intro} Dirac-Coulomb Hamiltonians and spectrum: main results}

We study the discrete spectrum of the so-called Dirac-Coulomb Hamiltonian for a relativistic spin-$\frac{1}{2}$ particle of mass $m$ and charge $-e<0$, moving in $\mathbb{R}^3$, and subject to the external scalar field due to the Coulomb interaction with a nucleus of atomic number $Z$ placed in the origin, that is, the operator
\begin{equation}\label{eq:Hformal}
 H\;:=\;-\ii c \hbar \,\bm{\alpha}\cdot \bm{\nabla}+\beta m c^2-\frac{\,cZ\alpha_\mathrm{f}\,}{|x|}\mathbbm{1}
\end{equation}
acting on the Hilbert space 
\begin{equation}
 \cH\;:=\;L^2(\mathbb{R}^3)\otimes\mathbb{C}^4\;\cong\;L^2(\mathbb{R}^3,\mathbb{C}^4,\ud x)\,,
\end{equation}
where $\hbar$ is Planck's constant, $c$ is the speed of light,
\begin{equation}
 \alpha_\mathrm{f}\;=\;\frac{e^2}{\hbar c}\;\approx\;\frac{1}{137}
\end{equation}
is the fine-structure constant, and $\bm{\alpha}\equiv(\alpha_1,\alpha_2,\alpha_3)$ and $\beta$ are the $4\times 4$ matrices
\begin{equation}
 \beta\;=\;\begin{pmatrix} 
            \mathbbm{1} & \mathbbm{O} \\
            \mathbbm{O} & -\mathbbm{1}
           \end{pmatrix}\,,\qquad
 \alpha_j\;=\;\begin{pmatrix}
               \mathbbm{O} & \sigma_j \\
               \sigma_j & \mathbbm{O}
              \end{pmatrix}\,,\qquad j\in\{1,2,3\}\,,
\end{equation}
having denoted by $\mathbbm{1}$ and $\mathbbm{O}$, respectively, the identity and the zero $2\times 2$ matrix, and by $\sigma_j$ the Pauli matrices
\begin{equation}
 \sigma_1\;=\;\begin{pmatrix}
               0 & 1 \\ 1 & 0
              \end{pmatrix}\,,\qquad
 \sigma_2\;=\;\begin{pmatrix}
               0 & -\ii \\ \ii & 0
              \end{pmatrix}\,,\qquad
 \sigma_3\;=\;\begin{pmatrix}
               1 & 0 \\ 0 & -1
              \end{pmatrix}\,.
\end{equation}

As well known\cite{Thaller-Dirac-1992}, if one initially defines $H$ on the natural domain $C^\infty_0(\mathbb{R}^3\!\setminus\!\{0\},\mathbb{C}^4)$, then $H$ has a unique self-adjoint realisation only when $Z\alpha_\mathrm{f}\leqslant\frac{\sqrt{3}}{\,2}$ (i.e., $Z\leqslant 118$, the `\emph{sub-critical}' regime), an infinite multiplicity of self-adjoint extensions arising for larger $Z$.

Let us set for convenience $\nu\equiv -Z\alpha_\mathrm{f}$ and adopt natural units $c=\hbar=m=e=1$. It is standard to exploit the symmetries of $H$ by passing to polar coordinates $x\equiv(r,\Omega)\in\mathbb{R}^+\!\times\mathbb{S}^2$, $r:=|x|$, for $x\in\mathbb{R}^3$, which induces the isomorphism
\begin{equation}
 L^2(\mathbb{R}^3,\mathbb{C}^4,\ud x)\;\cong\;L^2(\mathbb{R}^+,\ud r)\otimes L^2(\mathbb{S}^2,\mathbb{C}^4,\ud\Omega)\,,
\end{equation}
and then further decomposing
\begin{equation}
 L^2(\mathbb{S}^2,\mathbb{C}^4,\ud\Omega)\;\cong\;\bigoplus_{j\in\frac{1}{2}+\mathbb{N}}\;\;\;\bigoplus_{m_j=-j}^j\;\bigoplus_{\kappa_j=\pm(j+\frac{1}{2})}\mathcal{K}_{m_j,\kappa_j}
\end{equation}
in terms of the observables
\[
\begin{split}
 \bm{L}&=\bm{x}\times(-\ii\bm{\nabla})\,,\qquad\qquad\quad\;\bm{S}=-{\textstyle\frac{\ii}{4}}\,\bm{\alpha}\times\bm{\alpha}\,, \\
 \bm{J}&=\bm{L}+\bm{S}\equiv(J_1,J_2,J_3)\,,\quad K=\beta(2\bm{L}\cdot\bm{S}+\mathbbm{1})\,,
\end{split}
\]
where 
\begin{equation}
 \mathcal{K}_{m_j,\kappa_j}:=\;\mathrm{span}\{\Psi^+_{m_j,\kappa_j},\Psi^-_{m_j,\kappa_j}\}\;\cong\;\mathbb{C}^2
\end{equation}
and $\Psi^+_{m_j,\kappa_j}$ and $\Psi^-_{m_j,\kappa_j}$ are two orthonormal vectors in $\mathbb{C}^4$, and simultaneous eigenvectors of the observables $J^2\!\upharpoonright\! L^2(\mathbb{S}^2,\mathbb{C}^4,\ud\Omega)$, $J_3\!\upharpoonright\! L^2(\mathbb{S}^2,\mathbb{C}^4,\ud\Omega)$, and $K\!\upharpoonright\! L^2(\mathbb{S}^2,\mathbb{C}^4,\ud\Omega)$  with eigenvalue, respectively, $j(j+1)$, $m_j$, and $\kappa_j$. Each subspace
\begin{equation}\label{eq:def_space_H_mj_kj}
 \cH_{m_j,\kappa_j}\;:=\;L^2(\mathbb{R}^+,\ud r)\otimes\mathcal{K}_{m_j,\kappa_j}\;\cong\;L^2(\mathbb{R}^+,\mathbb{C}^2,\ud r)
\end{equation}
of $\cH$ is then a reducing subspace for $H$, which, through the overall isomorphism
\begin{equation}
U\;:\;L^2(\mathbb{R}^3,\mathbb{C}^4,\ud x)\;\xrightarrow[]{\cong}\;\bigoplus_{j\in\frac{1}{2}+\mathbb{N}}\;\;\;\bigoplus_{m_j=-j}^j\;\bigoplus_{\kappa_j=\pm(j+\frac{1}{2})}\cH_{m_j,\kappa_j}\,,
\end{equation}
is therefore unitarily equivalent to
\begin{equation}\label{eq:Dirac_operator_decomposition}
 UHU^*\;=\;\bigoplus_{j\in\frac{1}{2}+\mathbb{N}}\;\;\;\bigoplus_{m_j=-j}^j\;\bigoplus_{\kappa_j=\pm(j+\frac{1}{2})}\;h_{m_j,\kappa_j}\,,
\end{equation}
where
\begin{equation}\label{eq:def_operator_h_mj_kj}
\begin{split}
  h_{m_j,\kappa_j}\;&:=\;\begin{pmatrix}
                   1+\frac{\nu}{r} & -\frac{\ud}{\ud r}+\frac{\kappa_j}{r} \\
                   \frac{\ud}{\ud r}+\frac{\kappa_j}{r} & -1+\frac{\nu}{r}
                  \end{pmatrix}, \\
 \qquad\mathcal{D}(h_{m_j,\kappa_j})\;&:=\;C^\infty_0(\mathbb{R}^+)\otimes \mathcal{K}_{m_j,\kappa_j}\;\cong\;C^\infty_0(\mathbb{R}^+,\mathbb{C}^2)\,.
\end{split}
\end{equation}

By standard limit-point limit-circle arguments (see, e.g., Ref.~\cite[Chapter 6.B]{Weidmann-book1987}, and for details on the proof also Ref.~\cite[Section 2]{Gallone-AQM2017}), one sees that the operator $h_{m_j,\kappa_j}$ is essentially self-adjoint in the Hilbert space $\cH_{m_j,\kappa_j}$ if and only if
 \begin{equation}
  \nu^2\;\leqslant\;\kappa_j^2-\textstyle{\frac{1}{4}}\,,
 \end{equation}
 and it has deficiency indices $(1,1)$ otherwise. Thus, the operator $h_{\frac{1}{2},1}\oplus h_{\frac{1}{2},-1}\oplus h_{-\frac{1}{2},1}\oplus h_{-\frac{1}{2},-1}$, and hence $H$ itself, has deficiency indices $(4,4)$, and therefore a 16-real-parameter family of self-adjoint extensions.

Among the four relevant blocks the two ones with $k=1$ are identical, and so are the two ones with $k=-1$. The operator-theoretic analysis  of the self-adjoint extensions is completely analogous for each of the two possible signs of $k$. Moreover, for completeness, we include the treatment of both the electron and the corresponding positron, thus allowing the parameter $\nu$ to attain both positive and negative values for each of the two admissible values of $k$.

% With the above conventions in mind, let us focus henceforth on the block $h\equiv h_{\frac{1}{2},1}$. In the sub-critical regime $|\nu|\in(0,\frac{\sqrt{3}}{2}]$ the operator closure $\overline{h}$ is self-adjoint and is a very well studied Hamiltonian (the Dirac-Coulomb Hamiltonian for atoms with $Z\leqslant 118$) since the early times of quantum mechanics\cite{Thaller-Dirac-1992}. In particular,

%Let us then focus on the block $h\equiv h_{\frac{1}{2},1}$. 
In the sub-critical regime $|\nu|\in(0,\frac{\sqrt{3}}{2}]$ the operator closure $\overline{h}$, where $h$ denotes for a moment any of the four operators $h_{\pm\frac{1}{2},\pm 1}$, is self-adjoint and is a very well studied Hamiltonian (the Dirac-Coulomb Hamiltonian for atoms with $Z\leqslant 118$) since the early times of quantum mechanics\cite{Thaller-Dirac-1992}. In particular,
\begin{equation}
\begin{split}
 \sigma_{\mathrm{ess}}(h)\;&=\;(-\infty,-1]\cup[1,+\infty) \\
 \sigma_{\mathrm{disc}}(h)\;&=\;\{E_n\,|\,n\in \mathbb{N}_0\}\,.
\end{split}
\end{equation}
The eigenvalues $E_n$'s are given by Sommerfeld's celebrated fine-structure formula: for example, in the concrete case $\nu<0$, 
\begin{equation}\label{eq:Sommerfeld_formula}
 E_n\;=\;\Big(1+\frac{\nu^2}{(n+\sqrt{1-\nu^2})^2}\Big)^{\!-1/2}, \qquad\quad \nu <0
\end{equation}
(the general case is reported in formula \eqref{eq:EVEnk1} below).

It will be instructive in the following (Sec.~\ref{sec:Sommerfeld_formula}) to revisit the classical methods by which Sommerfeld's formula was derived. It is also worth noticing that in the non-relativistic limit $E_n$ reproduces the $(n+1)$-th energy level of the Schr\"{o}dinger-Coulomb problem: this is seen by reinstating for a moment physical units and constants, and computing
\[
 E_n-m c^2\;=\;m c^2\Big(\Big(1+\frac{\nu^2/c^2}{(n+\sqrt{1-\nu^2/c^2})^2}\Big)^{\!-1/2}-1\Big)\;\xrightarrow[]{\:c\to +\infty\:}\;-\frac{m\nu^2}{2(n+1)^2}\,.
\]

Evidently, Sommerfeld's formula \eqref{eq:Sommerfeld_formula} still yields \emph{real} eigenvalues for the \emph{larger} range $|\nu|\in(0,1)$ and only produces \emph{complex} (non-real) numbers when $|\nu|>1$. This has been since ever generically interpreted as the signature of the fact that when $|\nu|>1$, and hence $Z>137$, it is not possible any longer to make sense of $H$ as a Hamiltonian with bound states, thus obtaining an unstable model (the `$Z=137$ catastrophe').

Therefore, even beyond the regime of coupling $\nu$ in which $H$ is unambiguously defined as a self-adjoint operator, the remaining range $|\nu|\in(\frac{\sqrt{3}}{2},1)$ is of relevance because of the meaningfulness of formula \eqref{eq:Sommerfeld_formula} for bound states: this regime is usually referred to as the `\emph{critical regime}' and corresponds to ultra-heavy nuclei with atomic number $118\leqslant Z\leqslant 137$, possibly nuclei of elements whose discovery is expected in the near future (the last one to be discovered, the Oganesson ${}^{294}_{118}$Og, thus $Z=118$, was first synthesized in 2002 and formally named in 2016).

In fact, starting from the 1970's, and until present days, an intensive investigation has been carried on to identify and study a `\emph{distinguished}' realisation $H_D$ of $H$ in the critical regime, qualified by being the unique realisation whose domain is both contained in the form domain of the kinetic energy and in the form domain of the potential energy.\cite{Evans-1970,Weidmann-1971,Schmincke-1972-distinguished,Wust-1975,Nenciu-1976,Wust-1977,Klaus-Wust-1978,Landgren-Rejto-JMP1979,Arai-Yamada-RIMS-1983,Kato-1983,Xia-1999,Esteban-Loss-JMP2007,Voronov-Gitman-Tyutin-TMP2007,Arrizabalaga-JMP2011,Arrizabalaga-Duoandikoetxea-Vega_2012_JMP2013,Hogreve-2013_JPhysA,Esteban-Lewin-Sere-2017_DC-minmax-levels} As we shall re-derive later, formula \eqref{eq:Sommerfeld_formula} in the critical regime is nothing but the formula for the eigenvalue of such a distinguished extension, more precisely for the corresponding distinguished extension $h_D$ of $h$.

Much less investigated is instead the remaining family of self-adjoint extensions of $h$ and of their spectra.\cite{Voronov-Gitman-Tyutin-TMP2007,Hogreve-2013_JPhysA,MG_DiracCoulomb2017} Recently, in Ref.~\cite{MG_DiracCoulomb2017}, we produced a novel classification of the whole family of extensions of $h$ based on the so-called Kre{\u\i}n-Vi\v{s}ik-Birman\cite{GMO-KVB2017} and Grubb\cite{Grubb-1968}
%,Grubb-DistributionsAndOperators-2009} 
extension theory, as opposite to the previous classifications\cite{Voronov-Gitman-Tyutin-TMP2007,Hogreve-2013_JPhysA,Cassano-Pizzichillo-2017} based on the classical von Neumann theory. In this respect, Ref.~\cite{Cassano-Pizzichillo-2017} deals also with generic potentials $V(x)$ with local Coulomb singularity $|x|^{-1}$.

Let us briefly summarise our previous findings (see Ref.~\cite[Sec.~2]{MG_DiracCoulomb2017}).%, that for concreteness of the presentation we report for the sector $k=1$. 
We shall work in the critical regime $|\nu|\in\textstyle{(\frac{\sqrt{3}}{2},1)}$, whence %%%QUI TOLTO IL MODULO A NU
\begin{equation}\label{eq:Bnu}
 B\;:=\;\sqrt{1-\nu^2}\;\in\;\textstyle{(0,\frac{1}{2})}\,.
\end{equation}
We introduce the differential operator
%\begin{equation}\label{eq:diff_op_h}
% \widetilde{h}\;:=\;\begin{pmatrix}
%                   1+\frac{\nu}{r} & -\frac{\ud}{\ud r}+\frac{1}{r} \\
%                   \frac{\ud}{\ud r}+\frac{1}{r} & -1+\frac{\nu}{r}
%                  \end{pmatrix}
%\end{equation}
\begin{equation}\label{eq:diff_op_h}
 \widetilde{h}\;:=\;\begin{pmatrix}
                   1+\frac{\nu}{r} & -\frac{\ud}{\ud r}+\frac{k}{r} \\
                   \frac{\ud}{\ud r}+\frac{k}{r} & -1+\frac{\nu}{r}
                  \end{pmatrix}
\end{equation}
on `\emph{spinor}' functions of the form $f(x)\equiv\begin{pmatrix} f^+(x) \\ f^-(x)\end{pmatrix}$. The densely defined and symmetric operator on the Hilbert space $L^2(\mathbb{R}^+,\mathbb{C}^2)$ defined by
\begin{equation}\label{eq:def_h}
\mathcal{D}(h)\;:=\;C^\infty_0(\mathbb{R}^+,\mathbb{C}^2)\,,\qquad hf\;:=\;\widetilde{h}f
\end{equation}
has adjoint given by
\begin{equation}
\mathcal{D}(h^*)\;=\;\{\psi\in L^2(\mathbb{R}^+,\mathbb{C}^2)\,|\,\widetilde{h}\,\psi \in L^2(\mathbb{R}^+,\mathbb{C}^2)\}\,\qquad h^*\psi\;=\;\widetilde{h}\,\psi\,.
\end{equation}
One has
\begin{equation}\label{eq:kernelSstar}
\begin{split}
 \ker S^*\;&=\;\mathrm{span}\{\Phi\} \\
 \Phi^{\pm}(r)\;&:=\;e^{-r}r^{-B}\big( \textstyle{\frac{\pm(k+\nu)+B}{k+\nu}}\,U_{-B,1-2B}(2r)-\textstyle{\frac{2rB}{k+\nu}}\,U_{1-B,2-2B}(2r)\big)\,,
\end{split}
\end{equation}
where $U_{a,b}(r)$ is the Tricomi function (see Ref.~\cite[Sec.~13.1.3]{Abramowitz-Stegun-1964}). $\Phi$ is analytic on $(0,+\infty)$ with asymptotics
\begin{equation}\label{eq:Phi_asymptotics}
 \begin{split}
  \Phi(r)\;&=\;r^{-B}\,\textstyle{\frac{\Gamma(2B)}{\Gamma(B)}}
  \begin{pmatrix}
   \;\frac{k+\nu+B}{k+\nu} \\
   -\frac{k+\nu-B}{k+\nu}
  \end{pmatrix}+
  \begin{pmatrix}
   q^+\! \\ q^-\!
  \end{pmatrix}r^B+O(r^{1-B})\quad\textrm{as }\;r\downarrow 0 \\
   \Phi(r)\;&=\;2^B\begin{pmatrix}
               1 \\ -1
              \end{pmatrix}e^{-r}
              (1+O(r^{-1}))\quad\,\textrm{as }\;r\to +\infty\,,
 \end{split}
\end{equation}
where
\begin{equation}\label{eq:def_qpm}
 q^\pm\;:=\;\textstyle\frac{4^B(\pm(k+\nu)-B)\Gamma(-2B)}{(k+\nu)\Gamma(-B)}\;\;(\neq 0)\,.
\end{equation}
We also introduce the constants
%\begin{equation}\label{eq:def_ppm}
% p^\pm\;:=\;q^\pm\cdot\textstyle{\frac{(1+\nu)^2\Gamma(-B)}{4^B(1+\nu+B)\Gamma(1-2B)}}\,\|\Phi\|^2_{L^2(\mathbb{R}^+,\mathbb{C}^2)}\;\;(\neq 0)\,.
%\end{equation}
\begin{equation}\label{eq:def_ppm}
 p^\pm\;:=\;q^\pm\cdot\textstyle{\frac{(k+\nu) \cos(B \pi)}{4^B B}}\,\|\Phi\|^2_{L^2(\mathbb{R}^+,\mathbb{C}^2)}\;\;(\neq 0)\,.
\end{equation}

Then the following holds.

\begin{thm}\label{thm:extensions}~

\begin{itemize}
 \item[(i)] Any function $g=\begin{pmatrix} g^+ \\ g^-\end{pmatrix}\in\mathcal{D}(h^*)$ satisfies the short-distance asymptotics
  \begin{equation}\label{eq:coeff_a_b_BIS}
  g(r)\;=\;g_0\, r^{-B}+g_1r^B+o(r^{1/2})\qquad\textrm{as }\;r\downarrow 0
 \end{equation}
 for some $g_0,g_1\in\mathbb{C}^2$ given by the (existing) limits
 \begin{equation}\label{eq:coeff_a_b}
 \begin{split}
  g_0\;&:=\;\lim_{r\downarrow 0} \,r^B g(r) \\
  g_1\;&:=\;\lim_{r\downarrow 0} \,r^{-B}(g(r)-g_0r^{-B})\,.
 \end{split}
 \end{equation}
 \item[(ii)] The self-adjoint extensions of the operator $h$ on $L^2(\mathbb{R}^+,\mathbb{C}^2)$ defined in \eqref{eq:def_h} constitute a one-parameter family $(h_{\beta})_{\beta\in\mathbb{R}\cup{\{\infty\}}}$ of restrictions of the adjoint operator $h^*$, each of which is given by
\begin{equation}\label{eq:Sbeta_bc}
 \begin{split}
  h_{\beta}\;&:=\;h^*\upharpoonright\mathcal{D}(h_{\beta}) \\
  \mathcal{D}(h_{\beta})\;&:=\;\Big\{g\in\mathcal{D}(S^*)\,\Big|\,\frac{g_1^+}{g_0^+}=c_\nu \beta+d_\nu\Big\}\,,
 \end{split}
\end{equation}
where
\begin{equation}\label{eq:defcd}
 \begin{split}
  c_{\nu,k}\;&=\;p^+{\textstyle\Big(\frac{\Gamma(2B)}{\Gamma(B)}\,\frac{k+\nu+B}{k+\nu}\Big)^{\!-1}} \\
  d_{\nu,k}\;&=\;q^+{\textstyle\Big(\frac{\Gamma(2B)}{\Gamma(B)}\,\frac{k+\nu+B}{k+\nu}\Big)^{\!-1}},
 \end{split}
\end{equation}
and $p^+$ and $q^+$ are given, respectively, by \eqref{eq:def_ppm} and \eqref{eq:def_qpm}.
\item[(iii)] The extension $h_D\;:=\;h_{\beta=\infty}$ is the unique (`distinguished') extension satisfying
 \begin{equation}\label{eq:SD_uniqueness_properties}
  \mathcal{D}(h_D)\subset H^{1/2}(\mathbb{R}^+,\mathbb{C}^2)\quad\textrm{or}\quad \mathcal{D}(h_D)\subset\mathcal{D}[r^{-1}]\,,
 \end{equation}
 where the latter is the form domain of the multiplication operator by $r^{-1}$ on each component of $L^2(\mathbb{R}^+,\mathbb{C}^2)$ (the space of `finite potential energy'). $h_D$ is invertible on $L^2(\mathbb{R}^+,\mathbb{C}^2)$ with everywhere defined and bounded inverse.
\item[(iv)] The operator $h_\beta$ is invertible on the whole $L^2(\mathbb{R}^+,\mathbb{C}^2)$ if and only if $\beta\neq 0$, in which case
 \begin{equation}\label{eq:Sbeta-1}
  h_\beta^{-1}\;=\;h_D^{-1}+\frac{1}{\,\beta\|\Phi\|^{2}}\:|\Phi\rangle\langle\Phi|\,.
 \end{equation}
\item[(v)] For each %invertible 
 extension $h_\beta$,
 \begin{equation}\label{eq:sigmaess}
  \sigma_{\mathrm{ess}}(h_\beta)\;=\;\sigma_{\mathrm{ess}}(h_D)\;=\;(-\infty,-1]\cup[1,+\infty)\,.
 \end{equation}
 \item[(vi)]  The gap in the spectrum $\sigma(h_\beta)$ around $E=0$ is at least the interval $(-E(\beta),E(\beta))$, where
 \begin{equation}\label{eq:Ebeta}
  E(\beta)\;:=\;\frac{|\beta|}{\,|\beta| \|h_D^{-1}\|+1\,}\,.
 \end{equation}
\end{itemize}
\end{thm}

%\begin{rem}
% A completely analogous theorem holds for the sector $k=-1$, that is, for the self-adjoint extensions of the operator now defined as
%\begin{equation}\label{eq:def_h_with_k_minus1}
%\begin{split}
% \mathcal{D}(h)\;&:=\;C^\infty_0(\mathbb{R}^+,\mathbb{C}^2) \\
% h\;&:=\;\begin{pmatrix}
%                   1+\frac{\nu}{r} & -\frac{\ud}{\ud r}-\frac{1}{r} \\
%                   \frac{\ud}{\ud r}-\frac{1}{r} & -1+\frac{\nu}{r}\,.
%                  \end{pmatrix}.
%\end{split}
%\end{equation}
%We omit the details of how formulas 
%\eqref{eq:kernelSstar}-\eqref{eq:def_ppm} and \eqref{eq:defcd} change accordingly.
%\end{rem}

Let us come now to the main object of this work. We aim at qualifying the spectra of the generic extension $h_\beta$, as compared to the known spectrum of the distinguished extension $h_D$. In fact, we observe that there is a gap in the literature between the well-established knowledge on the one hand that for critical couplings the Dirac-Coulomb Hamiltonian admits an infinite multiplicity of self-adjoint realisations, and the availability on the other hand of an eigenvalue formula for the distinguished extension only.

Our recent classification\cite{MG_DiracCoulomb2017} of the whole family of self-adjoint realisations of $h$ turns out to provide the appropriate scheme to fill this gap in.

First, the natural question arises why the `classical' methods for the determination of Sommerfeld's formula, mainly the ODE/truncation-of-series approach and the supersymmetric approach, did not determine other than the eigenvalues of the \emph{distinguished} extension. We address this point in Section \ref{sec:Sommerfeld_formula}, exhibiting the precise steps of such classical methods in which one naturally selects only the discrete spectrum of the distinguished (and in fact also of a `mirror' distinguished) realisation.

It actually turns out that there are no explicit alternatives: indeed, in the ODE approach to the differential eigenvalue problem the only alternative to truncating series is to deal with eigenfunctions expressed by infinite series, and imposing the eigenfunction with eigenvalue $E$ to belong to some domain $\mathcal{D}(h_\beta)$ does not produce a closed formula for $E$ any longer; on the other hand, in the supersymmetric approach the first order differential eigenvalue problem is studied by an auxiliary second order differential problem whose solutions only exhibit the boundary condition typical of the distinguished (or also of the `mirror' distinguished) extension, with no access to different boundary conditions.

Next, we address the issue of how the eigenvalue formula \eqref{eq:Sommerfeld_formula}, valid for $\beta=\infty$, gets modified for a generic extension parameter $\beta$. Our result is the following.

\begin{thm}\label{thm:spectrum-beta}
Let $k\in\{\pm 1\}$ and let $(h_\beta)_{\beta\in(-\infty,\infty]}$ be the family of self-adjoint realisations, in the critical regime $|\nu|\in(\frac{\sqrt{3}}{2},1)$ of the Dirac-Coulomb Hamiltonian $h$ defined in \eqref{eq:def_h}, according to the parametrisation given by Theorem \ref{thm:extensions}. The discrete spectrum of a generic realisation $h_\beta$ consists of the countable collection %%%QUI TOLTO MODULO DI NU
 \begin{equation}
  \sigma_{\mathrm{disc}}(h_\beta)\;=\;\big\{E_n^{(\beta)}\,|\,n\in\mathbb{N}_0\,, n\geqslant n_0\big\}\;\subset\;(-1,1)
 \end{equation}
 of eigenvalues $E_n^{(\beta)}$ which are all the possible roots, enumerated in decreasing order when $\nu>0$ and in increasing order when $\nu<0$, of the transcendental equation
 \begin{equation}\label{eq:fEn_formula}
  \mathfrak{F}_{\nu,k}(E_n^{(\beta)})\;=\;c_{\nu,k} \,\beta + d_{\nu,k}\,,
 \end{equation}
 where the constants $c_{\nu,k}$ and $d_{\nu,k}$ are given by \eqref{eq:defcd}, and
 \begin{equation}\label{eq:Fnu}
 \begin{split}
  \mathfrak{F}_{\nu,k}(E)\;&:=\; \big(2 \sqrt{1-E^2}\big)^{2\sqrt{1-\nu^2}}\;\frac{\Gamma(-2\sqrt{1-\nu^2})}{\Gamma(2\sqrt{1-\nu^2})}\;\frac{\nu\sqrt{\frac{1-E}{1+E}}+k-\sqrt{1-\nu^2}}{\nu\sqrt{\frac{1-E}{1+E}}+k+\sqrt{1-\nu^2}}\;\times \\
  &\qquad\qquad\qquad\times\frac{\Gamma\big(\frac{\nu E }{\sqrt{1-E^2}}+\sqrt{1-\nu^2}\big)}{\Gamma\big(\frac{\nu E}{\sqrt{1-E^2}}-\sqrt{1-\nu^2}\big)}\,.
 \end{split}
 \end{equation}
 The starting index of the enumeration is $n_0=0$ if $k$ and $\nu$ have the same sign, and $n_0=1$ otherwise.
%   \begin{equation}
%   \mathfrak{F}_\nu(E)\;:=\; (2 \sqrt{1-E^2})^{2B}\;\frac{\Gamma(-2B)}{\Gamma(2B)}\;\frac{\Gamma(\frac{\nu E }{\sqrt{1-E^2}}+B)}{\Gamma(\frac{\nu E}{\sqrt{1-E^2}}-B)}\;\frac{\nu\sqrt{\frac{1-E}{1+E}}+1-B}{\nu\sqrt{\frac{1-E}{1+E}}+1+B}
%  \end{equation}
\end{thm}

Equation \eqref{eq:fEn_formula} of Theorem \ref{thm:spectrum-beta}, that will be proved in Section \ref{sec:beta-spectrum}, provides the implicit formula for the eigenvalues of the generic extension $h_\beta$. 
A formula of the eigenfunctions corresponding to the eigenvalues $E^{(\beta)}_n$ is found in the proof of Theorem \ref{thm:spectrum-beta} -- see \eqref{eq:psi_eigenf} in Section \ref{sec:beta-spectrum}.

In particular, equation \eqref{eq:fEn_formula} contains Sommerfeld's formula %\eqref{eq:Sommerfeld_formula} 
for the distinguished extension of $h$, namely the extension with $\beta=\infty$. For a comparison with the existing literature, let us formulate the latter consequence for generic $k\in\{\pm 1\}$.

%\begin{cor}\label{cor:eigenvalues_distinguished}
% Under the assumptions of Theorem \ref{thm:spectrum-beta}, let $h_D$ be the distinguished (i.e., $\beta=\infty$) self-adjoint extension of $h$. Then the eigenvalues $(E_n)_{n\in\mathbb{N}_0}$ of $h_D$ are given by
% \begin{equation}\label{eq:EVEnk1}
% \begin{array}{ll}
%  \begin{split}
%  E_n\;&=\;-\,\mathrm{sign}(\nu)\,\Big(1+\frac{\nu^2}{(n+\sqrt{1-\nu^2})^2}\Big)^{\!-1/2} \\
%  &\qquad\quad n\in\{0,1,2,\dots\}
% \end{split} & \qquad\textrm{when }k=1\;\;\;{\;}
% \end{array}
% \end{equation}
%and by
% \begin{equation}\label{eq:EVEnkm1}
% \begin{array}{ll}
%  \begin{split}
%  E_n\;&=\;-\,\mathrm{sign}(\nu)\,\Big(1+\frac{\nu^2}{(n+\sqrt{1-\nu^2})^2}\Big)^{\!-1/2} \\
%  &\qquad\quad n\in\{1,2,\dots\} \\
%  E_0\;&=\; \sqrt{1-\nu^2}
% \end{split} & \qquad\textrm{when }k=-1\,.
% \end{array}
% \end{equation}
%\end{cor}

\begin{cor}\label{cor:eigenvalues_distinguished}
 Under the assumptions of Theorem \ref{thm:spectrum-beta}, let $h_D$ be the distinguished (i.e., $\beta=\infty$) self-adjoint extension of $h$. Then the eigenvalues $(E_n)_{n=n_0}^\infty$ of $h_D$ are given by
 \begin{equation}\label{eq:EVEnk1}
  E_n\;=\;-\,\mathrm{sign}(\nu)\,\Big(1+\frac{\nu^2}{(n+\sqrt{1-\nu^2})^2}\Big)^{\!-1/2}\,,
 \end{equation}
 the starting index of the enumeration being $n_0=0$ if $k$ and  $\nu$ have the same sign, and $n_0=1$ otherwise.
\end{cor}

The first five eigenvalues $E^{(\beta)}_0,\dots,E^{(\beta)}_4$ for generic $\beta$ are plotted in Figure \ref{fig:Evs} for the concrete case $k=1$, $\nu>0$. We obtained this plot by computing numerically the intersection points of the curve $E\mapsto \mathfrak{F}_{\nu,k}(E)$ with horizontal lines corresponding to various values of $c_{\nu,k}\,\beta+d_{\nu,k}$. In this case when $\beta>0$ all eigenvalues are strictly negative (and accumulate to $-1$),  whereas for a region of negative $\beta$'s the first eigenvalue is positive. As to be expected, $E^{(\beta)}_0=0$ only for $\beta=0$: this corresponds to the sole non-invertible extension.

\begin{figure}[h]
\begin{center}
\includegraphics[scale=0.43]{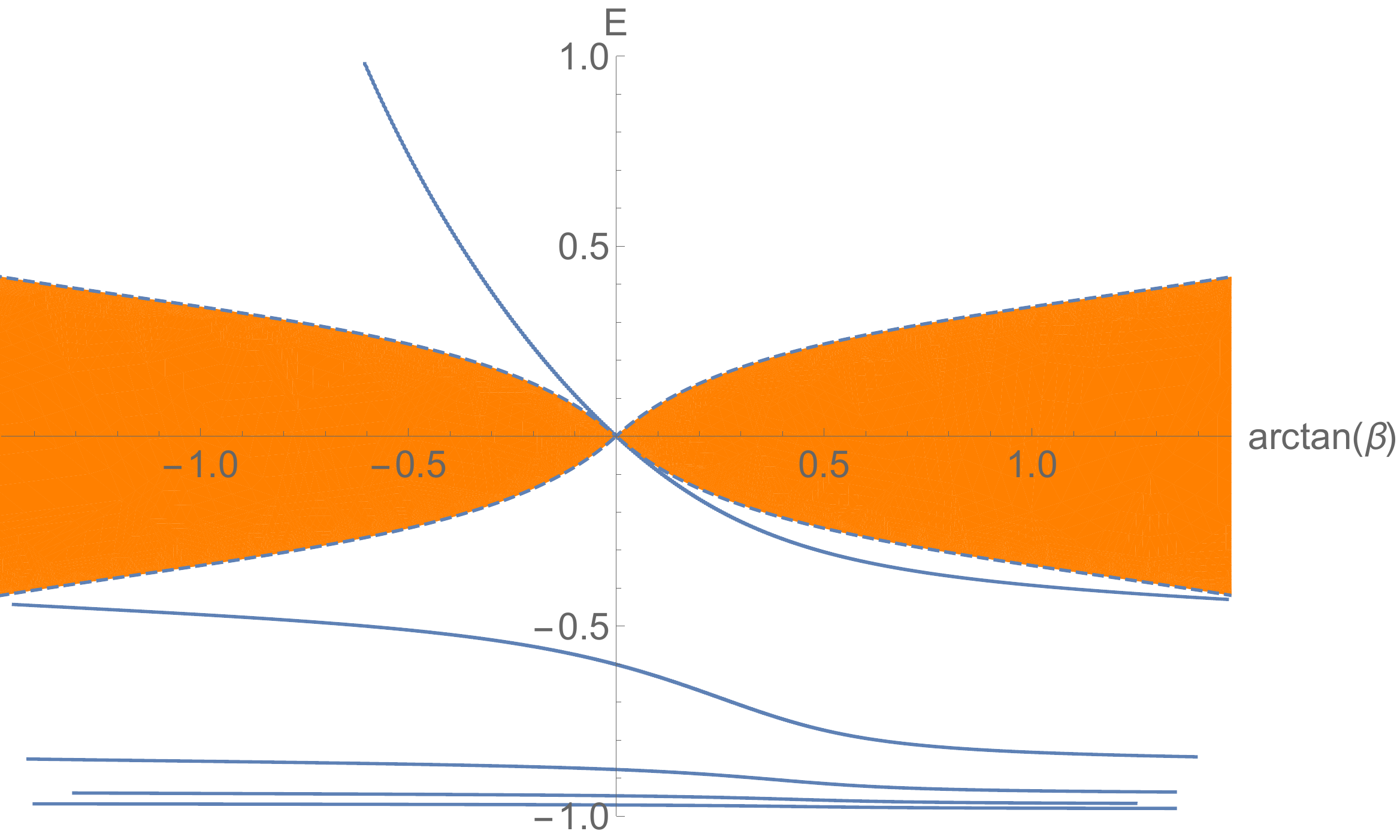}
\end{center}
\caption{Numerical computation of the eigenvalues $E^{(\beta)}_n$ as functions of $\beta$, for $k=1$ and $\nu=0.9$. The shaded area is the region $|E|<E(\beta)$, with $E(\beta)$ given by \eqref{eq:Ebeta}, and indicates the estimated gap in the spectrum around zero, according to Theorem \ref{thm:spectrum-beta}(vi).}\label{fig:Evs}
\end{figure}

\begin{figure}[h]
\begin{center}
\includegraphics[scale=0.3]{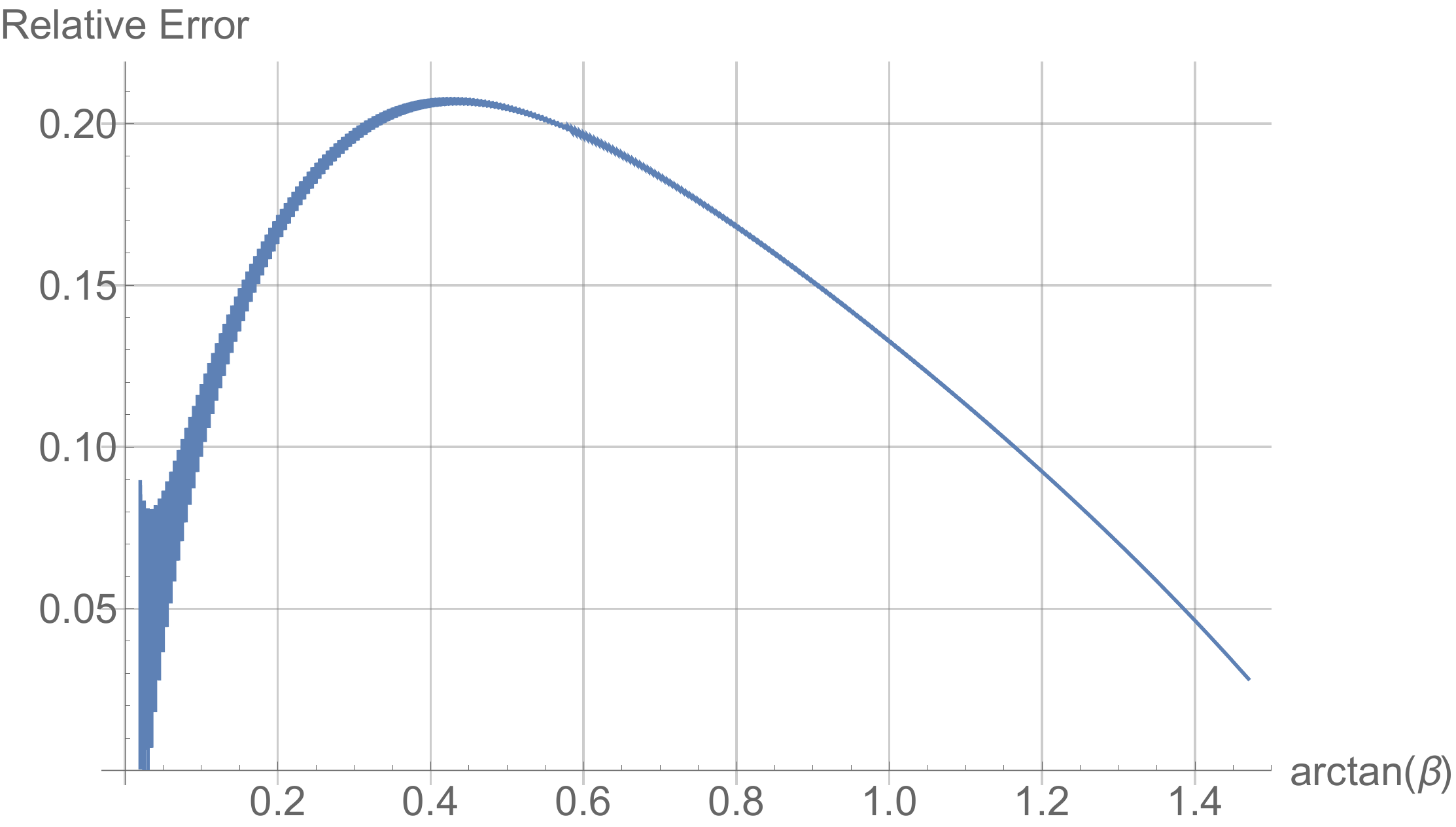}
\end{center}
\caption{Relative error on the estimate of the ground state energy for positive $\beta$, in the case $k=1$ and $\nu=0.9$. The worse relative error is reached for $\beta \sim 0.58$ and amounts to about $20 \%$.}
\label{fig:Errors}
\end{figure}

It follows from the detailed discussion of the behaviour of  $\mathfrak{F}_{\nu,k}(E)$ (in particular, of the vertical asymptotes of $\mathfrak{F}_{\nu,k}(E)$) which we are going to develop in Section \ref{sec:beta-spectrum}  that each $E^{(\beta)}_n$ is smooth and strictly monotone in $\beta$, and it moves with continuity from $\beta=(+\infty)^-$ to $\beta=(-\infty)^+$. This results in a typical \emph{fibred structure} of the union of all the discrete spectra $\sigma_{\mathrm{disc}}(h_\beta)$, with
\begin{equation}
 \bigcup_{\beta\in(-\infty,+\infty]}\big\{E_n^{(\beta)}\,|\,n\in\mathbb{N}_0\,, n\geqslant n_0\big\}\;=\;(-1,1)\,.
\end{equation}
This is a common phenomenon for the discrete spectra of one-parameter families of self-adjoint extensions of a given densely defined symmetric operator, where each extension is a rank-one perturbation, in the resolvent sense, of a reference extension: the complement of the essential spectrum, which is the same for all the extensions, is fibred by the union of all discrete spectra. We are already familiar with this phenomenon, to mention another physically relevant case, in the context of Hamiltonians of contact interaction, for example the two-body Hamiltonian\cite{albeverio-solvable} or the three-body `Ter-Martyrosyan--Skornyakov' Hamiltonian\cite{Minlos-Faddeev-1961-2}.

% 
% with
% \begin{equation}
% \begin{split}
%  \lim_{\beta\to +\infty}E^{(\beta)}_n\;& =\;\lim_{\beta\to -\infty}E^{(\beta)}_{n+1}\,,\qquad \lim_{\beta\downarrow-\frac{d_{\nu,k}}{c_{\nu,k}} }E^{(\beta)}_0\;=\;1\qquad\qquad(\nu>0) \\
%   \lim_{\beta\to +\infty}E^{(\beta)}_{n}\;& =\;\lim_{\beta\to -\infty}E^{(\beta)}_{n+1}\,,\qquad \lim_{\beta\downarrow-\frac{d_{\nu,k}}{c_{\nu,k}} }E^{(\beta)}_0\;=\;1\qquad\qquad(\nu<0)
% \end{split}
% \end{equation}

Let us conclude the presentation of our results with a comment on the accuracy of the estimate \eqref{eq:Ebeta} on the width of the spectral gap around zero for a generic extension $h_\beta$, estimate that we determined recently in Ref.~\cite{MG_DiracCoulomb2017}. Let us choose for concreteness $k=1$ and $\nu>0$: the estimated gap in this case is superimposed in Figure \ref{fig:Evs} and turns out to be asymptotically exact for $\beta\to 0$ and $\beta\to +\infty$, and reasonably precise in between. Owing to Corollary \ref{cor:eigenvalues_distinguished} we can now write
\begin{equation}\label{eq:normahdm1}
 \|h_D^{-1}\|\;=\;B^{-1}\;=\;(1-\nu^2)^{-\frac{1}{2}}\,.
\end{equation}
Thus, from \eqref{eq:Ebeta} and \eqref{eq:normahdm1} we conclude that 
\begin{equation}
 \mathcal{E}_0^{(\beta)}\;:=\;-\frac{\beta}{1+ \beta (1-\nu^2)^{-\frac{1}{2}}}
\end{equation}
provides a good estimate (from below) of the \emph{otherwise not explicitly computable} ground state $E_0^{(\beta))}$ of the generic self-adjoint extension $h_\beta$.

\section{Sommerfeld's eigenvalue formula revisited and spectrum of $h_D$}\label{sec:Sommerfeld_formula}

Prior to addressing the study of the discrete spectrum of the generic self-adjoint realisation $h_\beta$ (the essential spectrum being given by \eqref{eq:sigmaess}), it is instructive to revisit the two main methods by which Sommerfeld's formula has been known since long  for the eigenvalue problem of the differential operator $\widetilde{h}$ given by \eqref{eq:diff_op_h}, which will be the object of this Section.

The material is undoubtedly classical, and standard references will be provided below. Our perspective here is to highlight how such standard methods for the determination of the eigenvalues of $\widetilde{h}$ actually select the discrete spectrum of the distinguished realisation $h_D$ or of a `mirror' distinguished one, and as such are not applicable to the other realisations of $\widetilde{h}$.

In the next Section we shall indeed discuss how Sommerfeld's formula and its actual derivation gets modified for a generic extension $h_\beta$.

For concreteness, let us assume throughout this Section that $k=1$ and $\nu>0$. We therefore consider the eigenvalue problem
\begin{equation}\label{eq:EigenvalueEq}
h_\beta  \psi\;=\;E\psi\,,\qquad\psi\in\mathcal{D}(h_\beta)\,,\qquad E\in(-1,1)
\end{equation}
where $h$ given by \eqref{eq:def_h}, and hence the differential problem $\widetilde{h}\psi=E\psi$ with $\widetilde{h}$ given by \eqref{eq:diff_op_h}.

\subsection{The eigenvalue problem by means of truncation of asymptotic series}\label{sec:ODEmethods}

The historically first approach (see, e.g., Sec.~14 of Ref.~\cite{Bethe-Slapeter-1957}) for the determination of the eigenvalues of the Dirac-Coulomb Hamiltonian is based on ODE methods.

By direct inspection it is seen that the two linearly independent solutions to $\widetilde{h}\psi=E\psi$  have large-$r$ asymptotics $e^{r\sqrt{1-E^2}}$ and $e^{-r\sqrt{1-E^2}}$, only the second one being square-integrable and hence admissible. This suggests the natural re-scaling $\psi\mapsto U\psi=:\phi$ defined by
\begin{equation}\label{eq:rescaling_map_U}
 (U\psi)(\rho)\;:=\;{\textstyle{\frac{1}{\sqrt{2}(1-E^2)^{1/4}}}}\,\exp\big({\textstyle\frac{\rho}{2\sqrt{1-E^2}}}\big)\,\psi\big({\textstyle{\frac{\rho}{2\sqrt{1-E^2}}}}\big)\,,
\end{equation}
which induces the unitary operator $U: L^2(\mathbb{R}^+, \mathbb{C}^2,\ud  r) \to L^2(\mathbb{R}^+,\mathbb{C}^2,e^{-\rho} \, \ud \rho)$ and yields the unitarily equivalent problem
\begin{equation}\label{eq:EigenvalueEq1}
 U(h_\beta -E\mathbbm{1})U^{-1}\phi\;=\;0\,,\qquad\quad \phi\;:=\;U\psi\;\in\; U\mathcal{D}(h_\beta)\,,
\end{equation}
where
\begin{equation}\label{eq:EigenvalueEq2}
U(h_\beta -E\mathbbm{1})U^{-1}\;=\; 2\sqrt{1-E^2} \begin{pmatrix}
\frac{1}{2} \sqrt{\frac{1-E}{1+E}} + \frac{\nu}{\rho} & \frac{1}{2}-\frac{\ud}{\ud\rho} + \frac{1}{\rho} \\
-\frac{1}{2}+\frac{\ud}{\ud\rho}+\frac{1}{\rho} & -\frac{1}{2} \sqrt{\frac{1+E}{1-E}}+ \frac{\nu}{\rho}
\end{pmatrix}.
\end{equation}

The operator \eqref{eq:EigenvalueEq2} has a pole of order one at $\rho=0$, implying that the differential equation \eqref{eq:EigenvalueEq1} can be recast as 
\begin{equation}\label{eq:DiffEq}
\rho \,\phi' \;=\; A(\rho)\, \phi
\end{equation}
with
\begin{equation}\label{eq:VectorFieldMatrix}
A(\rho)\;:=\;\begin{pmatrix}
-1 & -\nu \\
\nu & 1
\end{pmatrix}+
\frac{1}{2}
\begin{pmatrix}
1 &  \sqrt{\frac{1+E}{1-E}} \\
\sqrt{\frac{1-E}{1+E}} & 1
\end{pmatrix}\rho\,.
\end{equation}
In particular it is explicitly checked that $\rho\mapsto A(\rho)$ is holomorphic.

It turns out that the differential problem \eqref{eq:DiffEq}-\eqref{eq:VectorFieldMatrix} is suited for the following standard result in the theory of ordinary differential equations (see, e.g., Ref.~\cite{Wasow_asympt_expansions}, Theorems 5.1 and 5.4).

\begin{prop}\label{prop:Wasow}
Let $z\mapsto B(z)$ be a matrix-valued function whose entries are holomorphic at $z=0$ and whose Taylor series 
$B(z)=\sum_{j=0}^\infty B_j z^j$, say, of radius of convergence $r_B$, has the zero-th component $B_0$ diagonal and with eigenvalues that do not differ by integers. Then there exists a holomorphic matrix-valued function $z\mapsto P(z)$ whose Taylor series $P(z)=\sum_{j=0}^\infty P_j z^j$ converges for $|z|<r_B$ and has zero-th component $P_0=\mathbbm{1}$, such that the transformation
\begin{equation}\label{eq:TransformationDiffEq}
y(z)\;=\;P(z) f(z)
\end{equation}
reduces the differential equation
\begin{equation}\label{eq:GenDiffEq}
z y'(z)\;=\;B(z) y(z)
\end{equation}
to the form
\begin{equation}\label{eq:NormalizedDiffEq}
z f'(z) \;=\;B_0 f(z).
\end{equation}
\end{prop}

Proposition \ref{prop:Wasow} is indeed applicable to \eqref{eq:DiffEq}-\eqref{eq:VectorFieldMatrix} whenever $\nu\in(0,1)\!\setminus\!\{\frac{\sqrt{3}}{2}\}$ because in this case the matrix $A_0=A(0)$ is diagonalizable and its two distinct eigenvalues $\pm B = \pm \sqrt{1-\nu^2}$ do not differ by an integer (indeed, $2B \notin \mathbb{Z}$). (For the purpose of the discussion of this Section, we do not need to cover the exceptional case $\nu=\frac{\sqrt{3}}{2}$ which presents particular features -- see, e.g., Ref.~\cite{Esteban-Lewin-Sere-2017_DC-minmax-levels}.)

Let us discuss first the (more relevant) critical regime $\nu\in(\frac{\sqrt{3}}{2},1)$: the argument for the sub-critical values $\nu\in(0,\frac{\sqrt{3}}{2})$ is even simpler and will be discussed at the end of this Subsection.

Proposition \ref{prop:Wasow} implies at once that the \emph{general solution} to \eqref{eq:DiffEq}-\eqref{eq:VectorFieldMatrix} has the form
\begin{equation}
\phi(\rho)\;=\;G P(\rho) \left(\begin{matrix} \rho^B & 0 \\ 0 & \rho^{-B} \end{matrix}\right) \phi_0
\end{equation}
for some holomorphic matrix-valued $P(\rho)$ and some vector $\phi_0\in \mathbb{C}^2$, where $G$ is the matrix that diagonalises $A_0$. Component-wise,
\begin{align}
\label{eq:Series1}\phi^{+}(\rho)&\;=\; \sum_{j=0}^\infty a_j^{(B)} \rho^{B+j} + \sum_{j=0}^\infty a_j^{(-B)} \rho^{-B+j}\\
\label{eq:Series2}\phi^{-}(\rho)&\;=\; \sum_{j=0}^\infty b_j^{(B)} \rho^{B+j} + \sum_{j=0}^\infty b_j^{(-B)} \rho^{-B+j}
\end{align}
for suitable coefficients $a_j^{(B)},b_j^{(B)},a_j^{(-B)},b_j^{(-B)}\in\mathbb{C}$, $j\in\mathbb{N}_0$, that must satisfy the consistency relations obtained by plugging \eqref{eq:Series1}-\eqref{eq:Series2} into \eqref{eq:DiffEq}. In doing so, one recognises that $\rho^{B+j}$-powers and $\rho^{-B+j}$-powers never get multiplied among themselves, and moreover each type of powers only gets multiplied by $a_j$ or $b_j$ coefficient of the same type; the net result, when equating to zero the coefficients of each power in the identity $\rho \phi'(\rho)-A(\rho)\phi(\rho)=0$ is the \emph{double set} of recursive equations
\begin{align}
\label{eq:Eq_First}{\textstyle \frac{1}{2} \sqrt{\frac{1-E}{1+E}}} \,a_j^{(\pm B)} + \nu\, a_{j+1}^{(\pm B)} + {\textstyle \frac{1}{2} }\,b_j^{(\pm B)} + (-j\mp B) b_{j+1}^{(\pm B)} &\;=\;0 \\
\label{eq:Eq_Second}{\textstyle -\frac{1}{2}} a_j^{(\pm B)} +(j\pm B+2) \,a_{j+1}^{(\pm B)} - {\textstyle \frac{1}{2} \sqrt{\frac{1+E}{1-E}}}\, b_j^{(\pm B)} + \nu \,b_{j+1}^{(\pm B)} &\;=\;0 \\
\label{eq:Eq_Zeroeth}\nu\,a_0^{(\pm B)}-(\pm B-1)\,b_0^{(\pm B)} &\;=\;0\,,
\end{align}
that is, the upper signs for the $B$-part and the lower signs for the $-B$-part of \eqref{eq:Series1}-\eqref{eq:Series2}.
% \begin{align}
% \label{eq:Eq_First}{\textstyle \frac{1}{2} \sqrt{\frac{1-E}{1+E}}} \,a_j + \nu\, a_{j+1} + {\textstyle \frac{1}{2} }\,b_j + (-j-\gamma) b_{j+1} &\;=\;0 \\
% \label{eq:Eq_Second}{\textstyle -\frac{1}{2}} a_j +(j+\gamma+2) a_{j+1} - {\textstyle \frac{1}{2} \sqrt{\frac{1+E}{1-E}}}\, b_j + \nu \,b_{j+1} &\;=\;0
% \end{align}

The above recursive relations are conveniently re-written in a more manageable form upon introducing $\alpha_j^{(\pm B)}$ and $\beta_j^{(\pm B)}$ through
\begin{equation}\label{eq:change}
 a_j^{(\pm B)}\;=\; \sqrt{1+E} \,(\alpha_j^{(\pm B)}+\beta_j^{(\pm B)})\,,\qquad b_j^{(\pm B)}\;=\;\sqrt{1-E} \, (\alpha_j^{(\pm B)}-\beta_j^{(\pm B)})\,,
\end{equation}
which yields
\begin{eqnarray}
\textstyle \big(\frac{\nu}{\sqrt{1-E^2}} +1 \big)\alpha_{j}^{(\pm B)}+\big( \frac{ E \nu}{\sqrt{1-E^2}} +j\pm B \big) \beta_{j}^{(\pm B)}\;&=&\;0  \label{eq:recursive01}\\
\textstyle\alpha_j^{(\pm B)}+\big(\frac{E \nu}{\sqrt{1-E^2}} -j-1\mp B \big)\, \alpha_{j+1}^{(\pm B)} + \big(\frac{\nu}{\sqrt{1- E^2}} - 1 \big) \,\beta_{j+1}^{(\pm B)} \;&=&\;0 \label{eq:recursive02}\\
\textstyle\big(\frac{\nu E}{\sqrt{1-E^2}} \mp B\big)\,\alpha_0^{(\pm B)}-\big(\frac{\nu}{\sqrt{1-E^2}} -1\big) \,\beta_0^{(\pm B)} \;&=& \;0\,. \label{eq:recursive00}
\end{eqnarray}
Now, plugging \eqref{eq:recursive01} into \eqref{eq:recursive02} yields 
\begin{equation}\label{eq:alpha_beta}
\alpha_{j+1}^{(\pm B)}\;=\;\frac{\frac{E \nu}{\sqrt{1-E^2}} + j \pm B +1 }{(j\pm B+1)^2-B^2}\,\alpha_j^{(\pm B)}\,.
\end{equation}

From \eqref{eq:alpha_beta} one sees that, unless $\alpha^{(\pm B)}_{j_0}=0$ for some $j_0$, in which case $\alpha^{(\pm B)}_{j}=0$ for all $j\geqslant j_0$, one has
\begin{equation}
\frac{\alpha_{j+1}^{(\pm B)}}{\alpha_j^{(\pm B)}}\;=\; j^{-1}+O(j^{-2})\qquad\textrm{as } j \to +\infty\,,
\end{equation}
implying that $\sum_j \alpha_j^{(\pm B)} \rho^{j}$ grows faster than $e^{\rho/2}$ at infinity and hence fails to belong to 
$L^2(\mathbb{R}^+, \mathbb{C}, e^{-\rho} \, \ud \rho)$. Through the transformation \eqref{eq:change} this implies that 
\begin{itemize}
 \item at least one among $\sum_j a_j^{(B)} \rho^{B+j}$ and $\sum_j b_j^{(B)} \rho^{B+j}$,
 \item and at least one among $\sum_j a_j^{(-B)} \rho^{-B+j}$ and $\sum_j b_j^{(-B)} \rho^{-B+j}$
\end{itemize}
are series that diverge faster than $e^{\rho/2}$. This poses the issue of admissibility (in particular, of the square-integrability) of the spinor-valued function $\phi$ given by \eqref{eq:Series1}-\eqref{eq:Series2}, for which the only possible affirmative answers are the following three.

\medskip

\textbf{\underline{First case}:} $\phi\in L^2(\mathbb{R}^+, \mathbb{C}^2, e^{-\rho} \, \ud \rho)$ because the $B$-series in  \eqref{eq:Series1} and the $B$-series in \eqref{eq:Series2} are actually truncated (i.e., polynomials), whereas the $(-B)$-series in  \eqref{eq:Series1} and the $(-B)$-series in \eqref{eq:Series2} vanish identically. This is obtained by imposing that $\alpha^{(B)}_{n+1}=0$ for some $n\in\mathbb{N}_0$ and that all the $a^{(-B)}_j$'s and $b^{(-B)}_j$'s vanish. Then \eqref{eq:alpha_beta} constrains $E$ to attain one of the values 
%\begin{equation}\label{eq:Truncation}
%E_n=  \Big(1+\frac{\nu^2}{(n+\sqrt{1-\nu^2})^2}\Big)^{-\frac{1}{2}}\qquad\textrm{for } \nu\lessgtr 0\,,\qquad n\in\mathbb{N}\,,
%\end{equation}
\begin{equation}\label{eq:Truncation}
E_n= - \Big(1+\frac{\nu^2}{(n+\sqrt{1-\nu^2})^2}\Big)^{-\frac{1}{2}} \qquad n\in\mathbb{N}\,.
\end{equation}
From \eqref{eq:recursive01} it is seen that the vanishing of $\alpha_{n+1}$ implies the vanishing of $\beta_j$ for all $j \geqslant n+2$ while, from \eqref{eq:recursive02}, one sees that $\beta_{n+1} \neq 0$. By direct inspection in \eqref{eq:recursive00} one sees that also $E_{n=0}$ given by \eqref{eq:Truncation} is an eigenvalue for which $\beta_0 \neq 0 $ and $\alpha_0 =0$ (it is crucial in this step that $\nu > 0$). Hence, for each value $E_n$, the corresponding $\phi$ has the form
\begin{equation}\label{eq:SommerfeldFormulaSeries}
\phi_{n}(\rho)\;=\;\rho^Be^{-\rho\sqrt{1-E_n^2}} \, \sum_{j=0}^{n+1}\begin{pmatrix} a_j^{(B)} \\ b_j^{(B)}\end{pmatrix}\rho^{j}\,,
\end{equation}
and through the inverse transformation $\psi=U^{-1}\phi$ of \eqref{eq:EigenvalueEq1} it is immediately recognised that $\psi$ satisfies the boundary condition \eqref{eq:Sbeta_bc} with $\beta=\infty$. This leads to the discrete spectrum of the \emph{distinguished} extension $h_D$: formula \eqref{eq:Truncation} is precisely the Sommerfeld's fine structure formula already introduced in \eqref{eq:Sommerfeld_formula}.

\medskip

\textbf{\underline{Second case}:} $\phi\in L^2(\mathbb{R}^+, \mathbb{C}^2, e^{-\rho} \, \ud \rho)$ because the $(-B)$-series in  \eqref{eq:Series1} and the $(-B)$-series in \eqref{eq:Series2} are finite polynomials, whereas the $B$-series in  \eqref{eq:Series1} and the $B$-series in \eqref{eq:Series2} vanish identically. This is obtained by imposing that $\alpha^{(-B)}_{n+1}=0$ for some $n\in\mathbb{N}_0$ and that all the $a^{(B)}_j$'s and $b^{(B)}_j$'s vanish. In this case \eqref{eq:alpha_beta} constrains $E$ to attain one of the values 
\begin{equation}\label{eq:anti-dist}
\begin{split}
 E_n&\;=\;- \displaystyle\Big(1+\frac{\nu^2}{(n-\sqrt{1-\nu^2})^2} \Big)^{-\frac{1}{2}} \qquad n\in\mathbb{N}\,,\\
 E_0&\;=\; B \,,
\end{split}
\end{equation}
the value $E_0$ being obtained by direct inspection in \eqref{eq:recursive00} analogously to what done for the analogous point in the previous case) and for each such value, the corresponding $\phi$ has the form
\begin{equation}
\phi_{n}(\rho)\;=\;  \rho^{-B}e^{-\rho\sqrt{1-E_n^2}} \, \sum_{j=0}^{n+1}\begin{pmatrix} a_j^{(-B)} \\ b_j^{(-B)}\end{pmatrix}\rho^{j}
\,.
\end{equation}
Through the inverse transformation $\psi=U^{-1}\phi$ of \eqref{eq:EigenvalueEq1} it is immediately recognised that $\psi$ satisfies the boundary condition \eqref{eq:Sbeta_bc} with
\begin{equation}\label{eq:mdbeta}
\beta\;=\; -\frac{d_\nu}{c_\nu}\,.
\end{equation}
This is another self-adjoint realisation of the Dirac-Coulomb Hamiltonian, different from $h_D$, which arises in this second case, where discussion mirrored the discussion of the first case for the distinguished extension. We shall refer to this realisation as the `\emph{mirror distinguished}' extension $h_{M\!D}$. We have thus found the discrete spectrum of $h_{M\!D}$, the eigenvalue formula \eqref{eq:anti-dist} providing the modification of Sommerfeld's formula for this Dirac-Coulomb Hamiltonian.

\medskip

It is crucial to observe at this point that the \emph{two eigenvalue formulas \eqref{eq:Truncation} and \eqref{eq:anti-dist} do not have any value in common}. As a consequence, even if combining together the truncation of the first case (in the $B$-series) and the truncation of the second case (in the $(-B)$-series) would produce a function $\phi$ that belongs to $L^2(\mathbb{R}^+, \mathbb{C}, e^{-\rho} \, \ud \rho)$, such $\phi$ could not correspond to any definite value $E$, i.e., $\phi$ could not be a solution to \eqref{eq:EigenvalueEq1}.

Truncation in \eqref{eq:Series1}-\eqref{eq:Series2} produces admissible solutions only of the form of truncated series of $B$-type or truncated series of $(-B)$-type. This explains why the only remaining case is the following.

\medskip

\textbf{\underline{Third case}:}  $\phi$ has the form \eqref{eq:Series1}-\eqref{eq:Series2} where \emph{both} component $\phi^+$ and $\phi^-$ contain two series that diverge faster than $e^{\rho/2}$ at infinity, whose sum however produces a compensation such that $\phi$ belongs to $L^2(\mathbb{R}^+, \mathbb{C}^2, e^{-\rho} \, \ud \rho)$. This yields then an admissible eigenfunction $\psi=U^{-1}\phi$ with eigenvalue $E$. Matching the coefficients of the expansion
\[
 \phi(\rho)\;=\;\rho^{-B}\begin{pmatrix} a_0^{(-B)} \\ b_0^{(-B)}  \end{pmatrix} +\rho^B\begin{pmatrix} a_0^{(B)} \\ b_0^{(B)}  \end{pmatrix} +\cdots\qquad\textrm{as } \rho \downarrow 0\,,
\]
through the transformation $\psi=U^{-1}\phi$, to the general boundary condition \eqref{eq:Sbeta_bc} indicates which domain $\mathcal{D}(h_\beta)$ the vector $\psi$ belongs to.

\medskip

Clearly, since in the third case above no truncation occurs in \eqref{eq:Series1}-\eqref{eq:Series2}, the recursive formulas for the coefficients are now of no use and it is not possible to infer from them any closed formula for the eigenvalues of the realisation $h_\beta$, $\beta\notin\{-\frac{d_\nu}{c_\nu},\infty\}$. In this sense, as announced at the beginning of this Section, the ODE methods discussed here only select the discrete spectrum (and a closed eigenvalue formula) for the distinguished extension $h_D$ and for the mirror distinguished extension $h_{M\!D}$.

To conclude this Subsection, we observe that in the sub-critical regime $\nu\in(0,\frac{\sqrt{3}}{2})$, i.e., $B\in (\frac{1}{2},1)$, the argument that led to the general form \eqref{eq:Series1}-\eqref{eq:Series2} is precisely the same, but of course in this regime $\rho^{-B}$ fails to be square-integrable near the origin, meaning that the whole $(-B)$-series in \eqref{eq:Series1}-\eqref{eq:Series2} must vanish identically. The only admissible solution is then that obtained with a truncation as in the first case, which leads again, as should be, to Sommerfeld's formula \eqref{eq:Truncation}.

\subsection{The eigenvalue problem by means of supersymmetric methods}

A second, by now classical\cite{Sukumar_1985,Grosse-1987,Cooper-Khare-Mustoa-Wipf-1988,Panahi-Bakhshi_2011}, approach to the determination of Sommerfeld's formula exploits the supersymmetric structure of the eigenvalue problem \eqref{eq:EigenvalueEq}.

By means of the bounded and invertible linear transformation $A:L^2(\mathbb{R}^+,\mathbb{C}^2)\to L^2(\mathbb{R}^+,\mathbb{C}^2)$ defined by
\begin{equation}
 A\,\xi\;:=\;
 \begin{pmatrix}
  -(1+B) & \nu \\
  \nu & -(1+B)
 \end{pmatrix}
 \begin{pmatrix}
  \xi^+ \\ \xi^-
 \end{pmatrix}
\end{equation}
it is convenient to turn the problem \eqref{eq:EigenvalueEq} into the form
\begin{equation}\label{eq:EigenvalueEq_v2}
\begin{split}
 0\;&=\;\sigma_2 \,A^{-1}\,\sigma_2\,(h_\beta-E\mathbbm{1})\,A\,\phi \\
 &=\;\left[\begin{pmatrix}
      0 & -\frac{\ud}{\ud r}+\frac{B}{r}+\frac{\nu E}{B} \\
      \frac{\ud}{\ud r}+\frac{B}{r}+\frac{\nu E}{B} & 0
     \end{pmatrix}
    -\begin{pmatrix}
      \frac{E}{B}-1 & 0 \\
      0 & \frac{E}{B}+1 
     \end{pmatrix}\right]\phi\,,
\end{split}
\end{equation} %Checked
having set
\begin{equation}\label{eq:psiAphi}
 \phi\;:=\;A^{-1}\psi\,. %;\in\;A^{-1}\mathcal{D}(h_\beta)\,.
\end{equation}

Next, in terms of the differential operators
\begin{equation}
 D^{\pm}\;:=\;\pm\frac{\ud}{\ud r}+\frac{B}{r}+\frac{\nu E}{B}
\end{equation}
acting on scalar functions, and of the differential operators
\begin{equation}
 Q\;:=\;\begin{pmatrix}
         \mathbbm{O} & D^- \\
         D^+ & \mathbbm{O}
        \end{pmatrix}\qquad\textrm{and}\qquad
 H\;:=\;Q^2\;=\;\begin{pmatrix}
         D^-D^+ & \mathbbm{O} \\
         \mathbbm{O} & D^+D^-
        \end{pmatrix}
\end{equation}
acting on spinor functions, equation \eqref{eq:EigenvalueEq_v2} reads
\begin{equation}\label{eq:ev_Qmatrixform}
 Q\phi\;=\;\begin{pmatrix}
      \frac{E}{B}-1 & 0 \\
      0 & \frac{E}{B}+1 
     \end{pmatrix}\phi\,,
\end{equation}
whence
\begin{equation}\label{eq:EigenvalueEq_v3}
 H\phi\;=\;Q^2\phi\;=\;Q\begin{pmatrix}
      \frac{E}{B}-1 & 0 \\
      0 & \frac{E}{B}+1 
     \end{pmatrix}\phi\;=
     \begin{pmatrix}
      \frac{E}{B}+1 & 0 \\
      0 & \frac{E}{B}-1 
     \end{pmatrix}Q\phi\;=\;{\textstyle(\frac{E^2}{B^2}-1)}\,\phi\,,
\end{equation}
equivalently,
\begin{equation}\label{eq:EigenvalueEq_v4}
 \begin{split}
  D^+D^-\phi^-\;&=\;{\textstyle(\frac{E^2}{B^2}-1)}\,\phi^- \\
  D^-D^+\phi^+\;&=\;{\textstyle(\frac{E^2}{B^2}-1)}\,\phi^+\,.
 \end{split}
\end{equation}

Equation \eqref{eq:EigenvalueEq_v3} or \eqref{eq:EigenvalueEq_v4} is the actual \emph{supersymmetric} form of \eqref{eq:EigenvalueEq}. The structure is indeed the same as for the triple $(\mathscr{H},\mathscr{P},\mathscr{Q})$, where (see, e.g., Ref.~\cite[Section 6.3]{Cycon-F-K-S-Schroedinger_ops} and Ref.~\cite[Section 5.1]{Thaller-Dirac-1992}), for some densely defined operator $D$ on $L^2(\mathbb{R}^+)$,
\begin{equation}
 \mathscr{Q}\;:=\;
 \begin{pmatrix}
  \mathbbm{O} & D^* \\
  D & \mathbbm{O}
 \end{pmatrix},\qquad
  \mathscr{P}\;:=\;
 \begin{pmatrix}
  \mathbbm{1} & \mathbbm{O} \\
  \mathbbm{O} & -\mathbbm{1}
 \end{pmatrix},\qquad
 \mathscr{H}\;:=\;\mathscr{Q}^2\;=\;
 \begin{pmatrix}
  D^*D & \mathbbm{O} \\
  \mathbbm{O} & DD^*
 \end{pmatrix}
\end{equation}
are self-adjoint operators on $L^2(\mathbb{R}^+)\oplus L^2(\mathbb{R}^+)\cong L^2(\mathbb{R}^+,\mathbb{C}^2)$ with the properties that $\mathscr{P}^2=\mathbbm{1}$, $\mathscr{P}\mathcal{D}(\mathscr{H})=\mathcal{D}(\mathscr{H})$, $\mathscr{P}\mathcal{D}(\mathscr{Q})=\mathcal{D}(\mathscr{Q})$, and $\{\mathscr{Q},\mathscr{P}\}=\mathbbm{O}$. Thus, $\mathscr{P}$ is an involution (the `grading operator'), $\mathscr{Q}$ is a `supercharge' with respect to such involution, and $\mathscr{H}$ is a Hamiltonian `with supersymmetry'. Moreover, standard spectral arguments show that the two spectra $\sigma(D^*D)$ and $\sigma(DD^*)$ with respect to $L^2(\mathbb{R}^+)$ lie both in $[0,+\infty)$ and coincide, and in particular the eigenvalues are the same, but for possibly the value zero.

In the present case we did not elaborate on the domain of $D^{\pm}$ when applied to $L^2(\mathbb{R}^+)$, however it is clear that the two operators are formally adjoint to each other. The fact that the eigenvalues of $D^+D^-$ and $D^-D^+$ relative to square-integrable eigenfunctions are non-negative follows from a trivial integration by parts; the fact that those such eigenvalues that are strictly positive are the same for both $D^+D^-$ and $D^-D^+$ is also an immediate algebraic consequence, for $D^-D^+f=\lambda f$ for $\lambda\neq 0$ implies that $D^+f\neq 0$ and $D^+D^-(D^+f)=\lambda(D^+f)$, the same then holding also when roles of $D^+$ and $D^-$ are exchanged.

The solutions $(E,\psi)$ to the problem \eqref{eq:EigenvalueEq}, with chosen realisation $h_\beta$, can be read out from \eqref{eq:EigenvalueEq_v3}-\eqref{eq:EigenvalueEq_v4}. Let us start with the `ground state' solutions, where `ground state' here is referred to the lowest possible eigenvalue of $H$, namely the value zero, and hence, because of \eqref{eq:EigenvalueEq_v3}, the smallest possible $|E|$ for the eigenvalue $E$ of the considered realisation $h_\beta$. First of all, the ground state energy $E_0$ must satisfy $E_0^2=B^2$, as follows from \eqref{eq:EigenvalueEq_v3}.

Out of the two possibilities, one is then to take $D^-\phi^-=0$ in \eqref{eq:EigenvalueEq_v4}, with $E=E_0$ to be determined, which is an ODE whose solutions are the multiples of
\[
 \phi^-(r)\;=\;r^{B}\,e^{\frac{\nu E_0}{B}r}\,.
\]

%nonononono... non capisco... quale deve essere l'asintotica all'infinito per phi? e manca la discussione sul segno di nu...
For such $\phi^-$ to be square-integrable, $\nu E_0<0$, thus $E_0=-B$ since $\nu>0$. Correspondingly, the second equation in \eqref{eq:EigenvalueEq_v4} is $D^-D^+\phi^+=0$ for some $\phi^+\in L^2(\mathbb{R}^+)$. This is equivalent to $D^+\phi^+=0$, thanks to the fact that $D^-$ is the formal adjoint of $D^+$. The latter ODE is solved by the multiples of $r^{-B}\,e^{-\frac{\nu E_0}{B}r}$, which is not square-integrable at infinity, whence $\phi^+=0$. Alternatively, one may argue that the corresponding $\phi^+$ to the above $\phi^-$ is read out directly from \eqref{eq:ev_Qmatrixform}: it must be (a multiple of)
\[
 ({\textstyle \frac{E_0}{B}-1})^{-1}(D^-\phi^-)(r)\;=\;({\textstyle \frac{E_0}{B}-1})^{-1}({\textstyle \frac{\ud}{\ud r}+\frac{B}{r}+\frac{\nu E_0}{B}})(r^{B}\,e^{\frac{\nu E_0}{B}r})
\]
and it must be square-integrable, which forces $\phi^+$ to be necessarily null, for the above function fails to be square-integrable at the origin.

We have thus found a solution $(E,\phi)$ to the problem \eqref{eq:EigenvalueEq_v4} with smallest possible $|E|$ and square-integrable $\phi$, namely the pair $(E_0,\phi_0)$ (up to multiples of $\phi_0$) given by
\begin{equation}\label{eq:gs_1_phi}
 E_0\;=\;-B\,,\qquad \phi_0(r)\;=\;r^{B}\,e^{\frac{\nu E_0}{B}r}
 \begin{pmatrix} 
 0 \\ 1
 \end{pmatrix}.
\end{equation}
Through \eqref{eq:Bnu} and the transformation \eqref{eq:psiAphi}, and in view of the classification \eqref{eq:defcd}, Theorem \ref{thm:extensions}(ii), we see that \eqref{eq:gs_1_phi} corresponds to the pair $(E_0,\psi_0)$ given by
\begin{equation}
 E_0\;=\;-\Big(1+\frac{\nu^2}{1-\nu^2}\Big)^{-\frac{1}{2}}\,,\qquad \psi_0(r)\;=\;r^{B}\,e^{\frac{\nu E_0}{B}r}
 \begin{pmatrix} 
 \nu \\ -(1+B)
 \end{pmatrix}\in\mathcal{D}(h_D)\,,
\end{equation}
which is the ground state solution to the initial eigenvalue problem \eqref{eq:EigenvalueEq} for $\beta=\infty$, and hence for the \emph{distinguished} self-adjoint realisation of the Dirac-Coulomb Hamiltonian.

By a completely analogous reasoning, the other possibility is to look for ground state solutions to \eqref{eq:EigenvalueEq_v4} with $D^+\phi^+=0$, and $E=E_0$ to be determined, an ODE solved by the multiples of
\[
 \phi^+(r)\;=\;r^{-B}\,e^{-\frac{\nu E_0}{B}r}\,,
\]
and such $\phi^+$ is only square-integrable if $E_0=B>0$. Correspondingly, the first equation in \eqref{eq:EigenvalueEq_v4} is $D^+D^-\phi^-=0$, equivalently, $D^-\phi^-=0$, which is solved by multiples of $r^{B}\,e^{\frac{\nu E_0}{B}r}$; the latter function failing to be square integrable at infinity, one thus ends up with the solution $(E_0,\phi_0)$ (up to multiples of $\phi_0$) given by
\begin{equation}\label{eq:gs_2_phi}
 E_0\;=\;B\,,\qquad \phi_0(r)\;=\;r^{-B}\,e^{-\frac{\nu E_0}{B}r}
 \begin{pmatrix} 
 1 \\ 0
 \end{pmatrix}.
\end{equation}
Thus, again using \eqref{eq:psiAphi}, and comparing the expansion
\[
 r^{-B}\,e^{-\frac{\nu E_0}{B}r}\;=\;r^{-B}-{\textstyle\frac{\nu E_0}{B}}\,r^{1-B}+o(r^{1-B})\qquad\textrm{as }r\downarrow 0
\]
with the general classification \eqref{eq:defcd}, whence now $g_0^+=1$, $g_1^+=0$, $c_\nu\beta+d_\nu=0$, we see that another ground state solution to \eqref{eq:EigenvalueEq} is the pair $(E_0,\psi_0)$ given by
\begin{equation}\label{eq:E0MD}
 E_0\;=\;\Big(1+\frac{\nu^2}{1-\nu^2}\Big)^{-\frac{1}{2}}\,,\qquad \psi_0(r)\;=\;r^{-B}\,e^{-\frac{\nu E_0}{B}r}
 \begin{pmatrix} 
 -(1+B) \\ \nu 
 \end{pmatrix}\in\mathcal{D}(h_{M\!D})\,,
\end{equation}
and this is the ground state solution for the \emph{mirror distinguished} ($\beta=-d_\nu/c_\nu$) self-adjoint realisation $h_{M\!D}$ already introduced in Subsection \ref{sec:ODEmethods}, formula \eqref{eq:mdbeta}. 

Significantly, no other realisations can be monitored through the supersymmetric scheme above, but those with $\beta=\infty$ or $\beta=-d_\nu/c_\nu$.

The  excited states too are determined within the supersymmetric scheme. Let
\begin{equation}\label{eq:Bplus}
 D_n^{\pm}\;:=\;\pm\frac{\ud}{\ud r}+\frac{B_n}{r}+\frac{\nu E}{B_n}\,,\qquad B_n:=B+n\,,\qquad n\in\mathbb{N}_0
\end{equation}
Clearly $B=B_0$, $D^\pm=D^\pm_0$. $D_n^+$ and $D_n^-$ are formally adjoint. From
\[
 \begin{split}
  D^\pm_n D^\mp_n\;&=\;-\frac{\ud^2}{\ud r^2}+\frac{B_n(B_n\mp 1)}{r^2}+\frac{2\nu E}{r}+\frac{\nu^2 E^2}{B_n^2}
 \end{split}
\]
one deduces
\begin{equation}\label{eq:Bshift}
 D^\pm_n D^\mp_n f\;=\;\big({\textstyle E^2(1+\frac{\nu^2}{B_n^2}})-1\big)f\qquad\Leftrightarrow\qquad-f''+{\textstyle \frac{B_n(B_n\mp 1)}{r^2}}f+{\textstyle\frac{2\nu E}{r}}f-E^2f\;=\;0\,.
\end{equation}
Thus, the equation in \eqref{eq:Bshift} with the lower signs is the same as the equation with the upper signs and with $B_n$ replaced by $B_{n+1}$. This is the basis for an iterative argument, as follows.

As a first step, as a consequence of \eqref{eq:Bshift}, the equation $D^-D^+\phi^+=(\frac{E^2}{B^2}-1)\phi^+$ of the problem \eqref{eq:EigenvalueEq_v3} is equivalent to $D_1^+D_1^-\phi^+=(E^2(1+\frac{\nu^2}{(B+1)^2})-1)\phi^+$, which can be regarded as the first scalar equation of
\begin{equation}\label{eq:susyproblem_step1}
\begin{pmatrix}
         D_1^-D_1^+ & \mathbbm{O} \\
         \mathbbm{O} & D_1^+D_1^-
        \end{pmatrix}
\begin{pmatrix}
 \xi_1^+ \\ \xi_1^-
\end{pmatrix}\;=\;\big({\textstyle E^2(1+\frac{\nu^2}{(B+1)^2})-1}\big)
\begin{pmatrix}
 \xi_1^+ \\ \xi_1^-
\end{pmatrix},\qquad \xi_1^-:=\phi^+\,.
\end{equation}
The ground state solution $(E_1,\xi_1^{(\textrm{gs})})$ to the new supersymmetric problem \eqref{eq:susyproblem_step1} is obtained in complete analogy to the argument that led to \eqref{eq:gs_1_phi}, whence
\begin{equation}\label{eq:Eqxi1}
 E_1\;=\;-\big({\textstyle 1+\frac{\nu^2}{(B+1)^2}} \big)^{\!-\frac{1}{2}}\,,\qquad \xi_1^{(\textrm{gs})}(r)\;=\;r^{B+1}e^{\frac{\nu E_1}{B+1}r}
 \begin{pmatrix} 0 \\ 1 \end{pmatrix}.
\end{equation}
(The other solution that one would find in complete analogy to the argument that led to \eqref{eq:gs_2_phi} is not square integrable.)
In turn, using $\phi^+=\xi_1^-$, \eqref{eq:Eqxi1} corresponds to a solution $\phi_1^+$ to the equation $D^-D^+\phi^+=(\frac{E_1^2}{B^2}-1)\phi^+$, and hence to a solution $(E_1,\phi_1)$ to the original problem  \eqref{eq:ev_Qmatrixform}-\eqref{eq:EigenvalueEq_v3},
given by
\begin{equation}
 \begin{split}
  E_1\;&=\;-\big({\textstyle 1+\frac{\nu^2}{(B+1)^2}} \big)^{\!-\frac{1}{2}}\;<\;E_0\;<\;0 \\
  \phi_1^+(r)\;&=\;r^{B+1}e^{\frac{\nu E_1}{B+1}r} \\
  \phi_1^-(r)\;&=\;({\textstyle \frac{E_1}{B}+1})^{-1}(D^+\phi_1^+)(r)\,.
 \end{split}
\end{equation}
Clearly $(D^+\phi_1^+)(r)\sim r^B$ as $r\downarrow 0$, and all together $\psi_1:=A\phi_1\in \mathcal{D}(h_D)$: thus, $(E_1,\psi_1)$ gives the first excited state for the eigenvalue problem \eqref{eq:EigenvalueEq} for the \emph{distinguished} realisation $h_D$.

The procedure is repeated for the iterated supersymmetric problems
\begin{equation}\label{eq:susyproblem_stepn}
\begin{split}
\begin{pmatrix}
         \mathbbm{O} & D_{n-1}^-   \\
         D_{n+1}^+ & \mathbbm{O}
        \end{pmatrix}\begin{pmatrix}
 \xi_{n-1}^+ \\ \xi_{n-1}^-
\end{pmatrix} \;&=\;\begin{pmatrix}
        E\sqrt{1+\frac{\nu^2}{B^2_{n-1}}}-1  & \mathbbm{O}   \\
        \mathbbm{O} & E\sqrt{1+\frac{\nu^2}{B^2_{n-1}}}+1
        \end{pmatrix}\begin{pmatrix}
 \xi_{n-1}^+ \\ \xi_{n-1}^-
\end{pmatrix}, \\
\begin{pmatrix}
         D_n^-D_n^+ & \mathbbm{O} \\
         \mathbbm{O} & D_n^+D_n^-
        \end{pmatrix}
\begin{pmatrix}
 \xi_n^+ \\ \xi_n^-
\end{pmatrix}\;&=\;\big({\textstyle E^2(1+\frac{\nu^2}{B_n^2})-1}\big)
\begin{pmatrix}
 \xi_n^+ \\ \xi_n^-
\end{pmatrix},\qquad \xi_n^-=\xi_{n-1}^+\,.
\end{split}
\end{equation}
The admissible ground state solution $(E_n,\xi_n^{(\textrm{gs})})$ for the second equation in \eqref{eq:susyproblem_stepn} is
\begin{equation}
 E_n\;=\;-\big({\textstyle 1+\frac{\nu^2}{(B+n)^2}} \big)^{\!-\frac{1}{2}}\,,\qquad \xi_n^{(\textrm{gs})}(r)\;=\;r^{B+n}e^{\frac{\nu E_n}{B+n}r}
 \begin{pmatrix} 0 \\ 1 \end{pmatrix};
\end{equation}
then, by the first equation in \eqref{eq:susyproblem_stepn} and the preceding iterations, the pair $(E_n,\phi_n)$ with
\begin{equation}
\phi_n\;:=\;\begin{pmatrix}
  D_{n-1}^+ \big(r^{B+n}e^{\frac{\nu E_n}{B+n}r}\big)\\
  D_0^+D_1^+\cdots D_{n-1}^+\big(r^{B+n}e^{\frac{\nu E_n}{B+n}r}\big)
\end{pmatrix}
\end{equation}
gives the $n$-th excited state solution to the original problem  \eqref{eq:ev_Qmatrixform}-\eqref{eq:EigenvalueEq_v3}. One immediately recognises that $\phi_n^\pm(r)\sim r^B$ as $r\downarrow 0$, whence $\psi_n:=A\phi_n\in \mathcal{D}(h_D)$: thus, $(E_n,\psi_n)$ gives the $n$-th excited state for the eigenvalue problem \eqref{eq:EigenvalueEq} for the \emph{distinguished} realisation $h_D$.

With the analysis above one reproduces all energy levels of Sommerfeld's formula
\begin{equation}
 E_n\;=\;-\big({\textstyle 1+\frac{\nu^2}{(n+\sqrt{1-\nu^2})^2}} \big)^{\!-\frac{1}{2}}\,,\qquad n\in\mathbb{N}_0
\end{equation}
and recognises that they all correspond to bound states for the distinguished realisation $h_D$ of the Dirac-Coulomb Hamiltonian.

By a completely symmetric iterative analysis which starts using
\begin{equation}
B_n:=B-n, \qquad n \in \mathbb{N}_0
\end{equation} 
instead of \eqref{eq:Bplus} and the same definitions for $D^\pm_n$ one sees that also the pairs $(E_n,\psi_n)$, with $\psi_n:=A\phi_n$ and
\begin{equation}\label{eq:as}
E_n\;:=\;-\big({\textstyle 1+\frac{\nu^2}{(n+\sqrt{1-\nu^2})^2}} \big)^{\!-\frac{1}{2}}\,,\qquad\phi_n\;:=\;\begin{pmatrix}
  D_0^+D_1^+\cdots D_{n-1}^+\big(r^{-B+n}e^{\frac{\nu E_n}{-B+n}r}\big)\\
  D_{n-1}^+ \big(r^{-B+n}e^{\frac{\nu E_n}{-B+n}r}\big)
\end{pmatrix},
\end{equation}
provide a complete set of solutions to the eigenvalue problem \eqref{eq:EigenvalueEq} for the \emph{mirror distinguished} realisation $h_{M\!D}$ ($\psi_n\in\mathcal{D}(h_\beta)$ for $\beta=-d_\nu/c_\nu$).

%\bigskip
%\bigskip
%
%PER MATTEO: Quest'ultimo capoverso l'ho aggiunto `by heart' avendo in mente le tue parole sulla anti-Sommerfeld formula e anti-Sommerfeld extension (che qui sto chiamando mirror distinguished, almeno temporaneamente), ma non sono sicuro che corrisponda mimiamente al vero. Sono certo che e' corretto concludere che \eqref{eq:E0MD} da' l'autovalore di ground state dell'estensione con $\beta=-d_\nu/c_\nu$, ma non so se \eqref{eq:as} ne e' davvero la generalizzazione a stati eccitati. In \eqref{eq:as} gli autovalori si accumulano a +1, e' qualcosa che e' confermato dal grafico numerico? Nella tua sezione dici che $E_n$ col + e' per i $\nu$ dell'altro segno, non parli di formula `anti-Sommerfeld' per un'altra estensione. Insomma, ragionaci su e fissa (o stralcia) questo punto.
%
%

% \begin{itemize}
%  \item Application of the ODE general theorem; use of the analiticity (double series); a continuum of $E$-eigenvalues and Tricomi's; Ansatz $b_n\equiv 0$ and truncation on the $a_n$'s yielding the the eigenvalue formula for the distinguished (or for the unique self-adjoint realisation in the sub-critical regime); Ansatz $a_n\equiv 0$ and truncation on the $b_n$'s yielding the anti-distinguished. References.
%  \item Alternative approach with super-symmetry. Again, distinguished + anti-distinguished.
% \end{itemize}

\section{Discrete spectrum of the generic extension}\label{sec:beta-spectrum}

In this Section we prove Theorem \ref{thm:spectrum-beta} and Corollary \ref{cor:eigenvalues_distinguished}.

For Theorem \ref{thm:spectrum-beta} we study the eigenvalue problem for $h_\beta$ in the form of the differential equation \eqref{eq:EigenvalueEq1}-\eqref{eq:EigenvalueEq2} already identified in Subsection \ref{sec:ODEmethods}. The key point is the intimate relation between the differential operator \eqref{eq:EigenvalueEq2} and the confluent hypergeometric equation. Exploiting such a relation yields, in the operator-theoretic language of Theorem \ref{thm:extensions}, the explicit expression for the eigenfunctions of the adjoint $h^*$ of $h$. Imposing further that such eigenfunctions satisfy the typical boundary condition for the $h_\beta$-extension brings eventually to the implicit eigenvalue formula \eqref{eq:fEn_formula}.

\begin{proof}[Proof of Theorem \ref{thm:spectrum-beta}]
Let us start from the differential problem \eqref{eq:EigenvalueEq}, re-written in the form \eqref{eq:EigenvalueEq1}-\eqref{eq:EigenvalueEq2}.%, now carrying on a unified notation for the self-adjoint extensions $h_\beta$ of the operator $h$ defined in \eqref{eq:def_h} when $k=1$ and in \eqref{eq:def_h_with_k_minus1} when $k=-1$. 

For a solution $\phi$ to \eqref{eq:EigenvalueEq1} with given $E\in(-1,1)$ we introduce, in analogy to \eqref{eq:change}, the two scalar functions $u_1$ and $u_2$ such that 
\begin{equation}\label{eq:def_u1u2}
 \begin{split}
  \phi^+&\;=\;\sqrt{1+E} \, (u_1+u_2) \\
  \phi^-&\;=\; \sqrt{1-E} \, (u_1-u_2)\,.
 \end{split}
\end{equation}
Plugging \eqref{eq:def_u1u2} into \eqref{eq:EigenvalueEq1}-\eqref{eq:EigenvalueEq2} yields
\begin{equation}\label{eq:u1u2eq}
 \begin{split}
  \textstyle u_2'+\Big(\frac{k}{\rho}+\frac{\nu}{\rho\,\sqrt{1-E^2}\,} \Big) u_1 +\frac{\nu E }{\rho\,\sqrt{1-E^2}\,} u_2 \;&=\; 0 \\
   \textstyle -u_1'+\Big(1+\frac{\nu E }{\rho\,\sqrt{1-E^2}\,} \Big) u_1+ \Big(\frac{\nu}{\rho\,\sqrt{1-E^2}\,}-\frac{k}{\rho} \Big) u_2 \;&=\;0\,,
 \end{split}
\end{equation}
and solving for $u_1$ in the first equation above and plugging it into the second equation gives a second order differential equation for $u_2$ which, re-written for the scalar function $v:=\rho^B u_2$, takes the form 
\begin{equation}\label{eq:hypergeom-v}
 \rho \,v''+(1-2B-\rho) \,v'-\Big(\frac{\nu E}{\sqrt{1-E^2}}-B \Big) v \; = \; 0\,.
\end{equation}

Equation \eqref{eq:hypergeom-v} is a confluent hypergeometric equation -- we refer, e.g., to Ref.~\cite[Chapter 13]{Abramowitz-Stegun-1964} for its definition and for the properties that we are going to use here below. Out of the two linearly independent solutions to \eqref{eq:hypergeom-v}, the Kummer function $M_{a,b}(\rho)$ and the Tricomi function $U_{a,b}(\rho)$ with parameters 
\begin{equation}\label{eq:parameters_ab}
 \textstyle a\;=\;\frac{\nu E}{\sqrt{1-E^2}}-B\,,\qquad b\;=\;1-2B\,,
\end{equation}
only the latter belongs to $L^2(\mathbb{R}^+,\mathbb{C}, e^{-\rho}\ud\rho)$, for
\[
 \begin{array}{l}
  M_{a,b}(\rho)\;=\;e^r\,\frac{\,r^{a-b}}{\Gamma(a)}(1+O(r^{-1})) \\
  \,\,U_{a,b}(\rho)\;=\;r^{-a}(1+O(r^{-1}))
 \end{array}\qquad \textrm{as } r\to +\infty\,.
\]
With $u_2=\rho^{-B} v=\rho^{-B} U_{a,b}(\rho)$, and with $u_1$ determined by \eqref{eq:u1u2eq} and the property
\[
 U'_{a,b}(\rho)\;=\;-a\, U_{a+1,b+1}(\rho)\,,
\]
we reconstruct the solution $\phi$ by means of \eqref{eq:def_u1u2} and we find
\begin{equation}
 \phi^{\pm}(\rho)\;=\;\frac{\rho^{-B}}{\,k+\frac{\nu}{\sqrt{1-E^2}}} \textstyle \big( \big(B\pm\nu\sqrt{\frac{1-E}{1+E}}\pm k \big) \,U_{a,b}(\rho)+a\, \rho\, U_{a+1,b+1}(\rho) \big)\,.
\end{equation}

Correspondingly, the solution $\psi=U^{-1}\phi$ to the differential problem $\widetilde{h}\psi=E\psi$, where $U:L^2(\mathbb{R}^+,\mathbb{C}^2,\ud r)\to L^2(\mathbb{R}^+,\mathbb{C}^2,e^{-\rho}\ud \rho)$ is the unitary map \eqref{eq:rescaling_map_U}, takes the form
\begin{equation}\label{eq:psi_eigenf}
\begin{split}
 \psi^{\pm}(r)\;&=\;\frac{(2r\sqrt{1-E^2})^{-B}\,e^{-r\sqrt{1-E^2}}}{\,k+\frac{\nu}{\sqrt{1-E^2}}} \Big( \textstyle \sqrt{1\pm E}\,\big(B\pm\nu\sqrt{\frac{1-E}{1+E}}\pm k \big) \,U_{a,b}(2r\sqrt{1-E^2}) \\
 &\qquad\qquad\qquad\qquad\qquad\quad+2ar\sqrt{1-E^2}\, U_{a+1,b+1}(2r\sqrt{1-E^2}) \Big)\,.
\end{split}
\end{equation}
From the above expression we deduce the asymptotics
\begin{equation}\label{eq:psiplus}
\begin{split}
\!\!\!\!\!\!\!\psi^+(r) \;&=\; \textstyle\frac{\Gamma(1-b)}{\Gamma(1+a-b)} \big(B+\nu\sqrt{\frac{1-E}{1+E}}  + k \big) r^{-B} + \frac{\Gamma(b-1)}{\Gamma(a)} (2\sqrt{1-E^2})^{2B} \big(\nu \sqrt{\frac{1-E}{1+E}} + k - B \big) r^B  \\
&\qquad + o(r^{1/2})\qquad\textrm{as }r\downarrow 0\,.
\end{split}
\end{equation}

Since $\widetilde{h}\psi=E\psi\in L^2(\mathbb{R}^+,\mathbb{C}^2,\ud r)$, then $\psi\in\mathcal{D}(h^*)$. Therefore, comparing \eqref{eq:parameters_ab} and \eqref{eq:psiplus} above with the general formulas \eqref{eq:coeff_a_b_BIS}-\eqref{eq:coeff_a_b} of Theorem \ref{thm:extensions}, we read out the coefficients
\begin{equation}\label{eq:g0g1_final}
 \begin{split}
  g_0^+ \;&=\; \textstyle \frac{\Gamma(2 B)}{\Gamma(\frac{\nu E}{\sqrt{1-E^2}}+B)} \big(\nu\sqrt{\frac{1-E}{1+E}} + k+B \big) \\
g_1^+ \;&=\; \textstyle(2 \sqrt{1-E^2})^{2B}\,\frac{\Gamma(-2 B)}{\Gamma(\frac{\nu E }{\sqrt{1-E^2}} - B)} \, \big(\nu \sqrt{\frac{1-E}{1+E}}  + k - B \big)
 \end{split}
\end{equation}
of the small-$r$ expansion $\psi(r)=g_0 r^{-B}+g_1 r^B+o(r^{1/2})$.

We are now in the condition to apply our classification formula  \eqref{eq:Sbeta_bc} to such $\psi$. Upon setting
  \begin{equation}\label{eq:setting_of_F}
  \mathfrak{F}_{\nu,k}(E)\;:=\;\frac{g_1^+}{g_0^+}\;=\; (2 \sqrt{1-E^2})^{2B}\;\frac{\Gamma(-2B)}{\Gamma(2B)}\;\frac{\Gamma(\frac{\nu E }{\sqrt{1-E^2}}+B)}{\Gamma(\frac{\nu E}{\sqrt{1-E^2}}-B)}\;\frac{\nu\sqrt{\frac{1-E}{1+E}}+k-B}{\nu\sqrt{\frac{1-E}{1+E}}+k+B}
 \end{equation}
we deduce from \eqref{eq:g0g1_final} and \eqref{eq:Sbeta_bc} that the function $\psi\in\mathcal{D}(h^*)$ determined so far actually belongs to $\mathcal{D}(h_\beta)$, and therefore is a solution to $h_\beta\psi=E\psi$, if and only if $E$ satisfies
\begin{equation}\label{eq:fE_beta_equation}
 \mathfrak{F}_{\nu,k}(E)\;=\;c_{\nu,k}\, \beta + d_{\nu,k}\,,
\end{equation}
which then proves \eqref{eq:fEn_formula}.

It is straightforward to deduce from the properties of the $\Gamma$-function that the map $(-1,1)\ni E\mapsto \mathfrak{F}_{\nu,k}(E)$ has the following features. $\mathfrak{F}_{\nu,k}$ has vertical asymptotes corresponding to the roots of 
\begin{equation}\label{eq:roots}
 \frac{{ \textstyle\Gamma\big(\frac{\nu E }{\sqrt{1-E^2}}+B}\big)}{{ \textstyle\Gamma\big(\frac{\nu E }{\sqrt{1-E^2}}-B}\big)}\times\frac{{\textstyle \nu\sqrt{\frac{1-E}{1+E}}+k-B}}{{\textstyle{\nu\sqrt{\frac{1-E}{1+E}}+k+B}}}\;=\;\infty\,.
\end{equation}
As we shall determine in detail working out equation \eqref{eq:roots} in the proof of Corollary \ref{cor:eigenvalues_distinguished}, such roots are indeed countably many and the corresponding asymptotes are located at the points $E=E_n$, with $E_n$ given by formula \eqref{eq:EVEnk1}. Therefore the asymptotes accumulate at $E=-1$ for $\nu>0$ and at $E=1$ for $\nu<0$. When $\nu>0$, in each interval $(E_{n+1},E_{n})$, as well as in the interval $(E_{n_0},1)$, $\mathfrak{F}_{\nu,k}$ is smooth and strictly monotone decreasing; the value $\mathfrak{F}_\nu(1)$ is finite and negative. When $\nu<0$ one has conversely that in each interval $(E_{n},E_{n+1})$, as well as in the interval $(-1,E_{n_0})$, $\mathfrak{F}_{\nu,k}$ is smooth and strictly monotone increasing.

Thus, the range of $\mathfrak{F}_{\nu,k}$ is the whole real line, which makes the equation \eqref{eq:fE_beta_equation} always solvable for any $\beta$, again with a countable collection of roots. This completes the proof.
\end{proof}

\begin{figure}[h!]
\begin{center}
\includegraphics[scale=0.35]{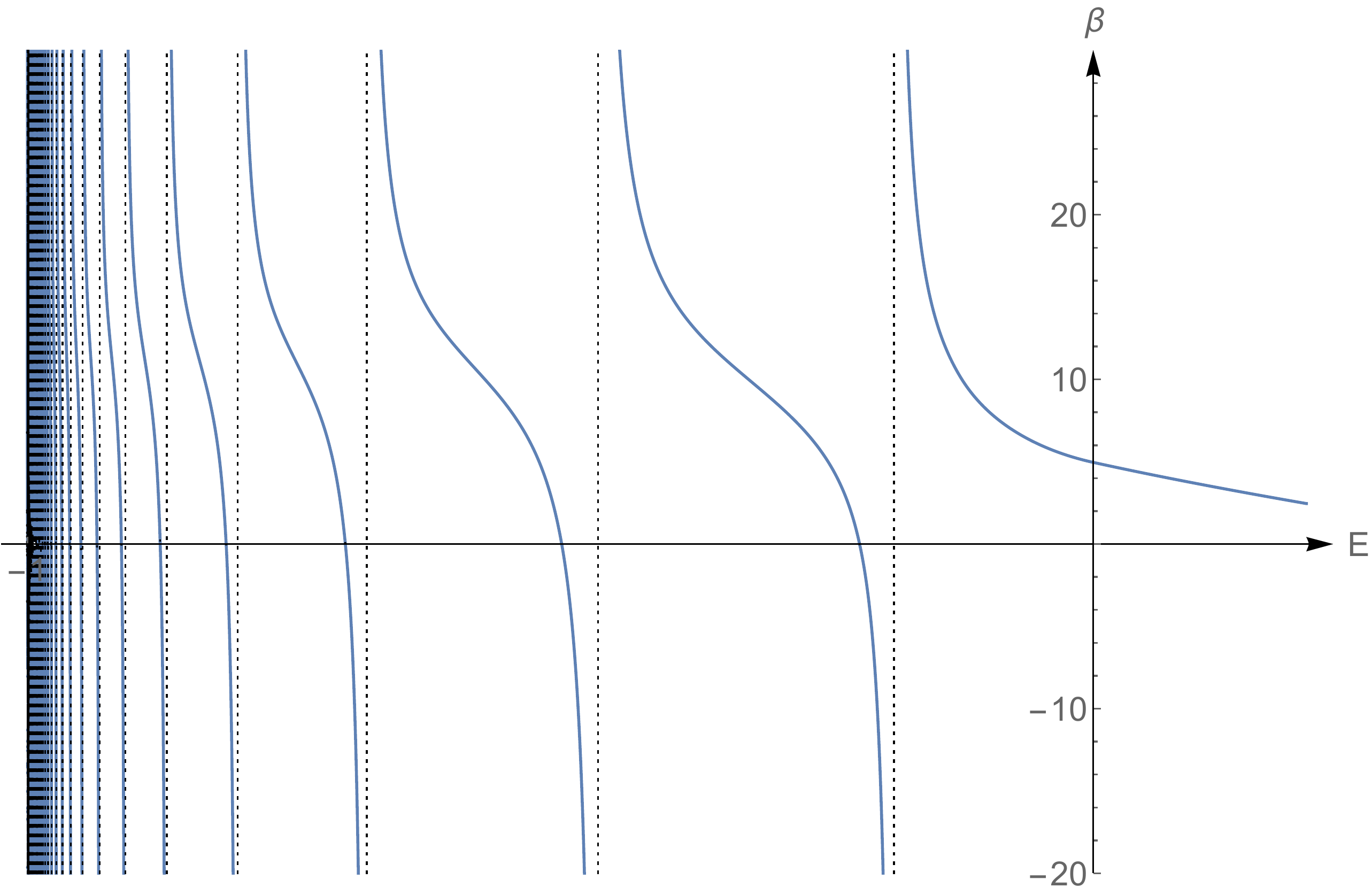}
\end{center}
\caption{Plot of $\mathfrak{F}_{\nu,k}(E)$ for $k=1$ and $\nu=0.9$ for $E \in (-1,0.3)$.}\label{fig:mathfrakF}
\end{figure}

The behaviour of $E\mapsto \mathfrak{F}_{\nu,k}(E)$ discussed above is illustrated in Figure \ref{fig:mathfrakF}  for $k=1$ and $\nu>0$. Observe that in this case the points $E_n$ where the vertical asymptotes are located at are all negative and $E_n\to -1$ as $n\to +\infty$. For
$\beta\in(-\infty,\mathfrak{F}_{\nu,k}(1))\cup(d_{\nu,k},+\infty)$ all such roots are strictly negative, whereas for $\beta\in (\mathfrak{F}_{\nu,k}(1),d_{\nu,k})$ the lowest root (and only that one) is strictly positive. As to be expected, $\mathfrak{F}_{\nu,k}(0)=d_{\nu,k}$, as one can easily see by comparing the value $\mathfrak{F}_{\nu,k}(0)$ obtained from \eqref{eq:setting_of_F} with the quantity $d_\nu$ given by \eqref{eq:def_qpm}/\eqref{eq:defcd}.

Let us now move to the derivation of Sommerfeld's formula from our general eigenvalue equation.

\begin{proof}[Proof of Corollary \ref{cor:eigenvalues_distinguished}]
 The goal is to determine the roots of $\mathfrak{F}_{\nu,k}(E)=\infty$, equivalently, the roots of  equation \eqref{eq:roots}. For each of the four factors
% \[
%  \begin{split}
%   P_\nu(E)\;&:=\; \textstyle\Gamma\big(\frac{\nu E }{\sqrt{1-E^2}}+B\big) \\
%   Q_{\nu,k}(E)\;&:=\; \textstyle \nu\sqrt{\frac{1-E}{1+E}}+k-B \\
%   R_{\nu,k}(E)\;&:=\;\textstyle\nu\sqrt{\frac{1-E}{1+E}}+k+B \\
%   S_{\nu}(E)\;&:= \;\textstyle\Gamma\big(\frac{\nu E }{\sqrt{1-E^2}}-B\big)
%  \end{split}
% \]
 \[
  \begin{split}
   P_\nu(E)\;&:=\; \textstyle\Gamma\big(\frac{\nu E }{\sqrt{1-E^2}}+B\big) \\
   Q_{\nu,k}(E)\;&:=\; \textstyle \nu\sqrt{\frac{1-E}{1+E}}+k-B \\
   R_{\nu,k}(E)\;&:=\;\textstyle\nu\sqrt{\frac{1-E}{1+E}}+k+B \\
   S_{\nu}(E)\;&:= \;\textstyle\Gamma\big(\frac{\nu E }{\sqrt{1-E^2}}-B\big)
  \end{split}
 \]

 in the l.h.s.~of \eqref{eq:roots} it is straightforward to find the following.
 \begin{itemize}
  \item $P_\nu(E)=\infty$ for $\frac{\nu E }{\sqrt{1-E^2}}+B=-n$, $n\in\mathbb{N}_0$, and hence for $E=-\mathrm{sign}(\nu)\,\mathcal{E}_n$ with 
  \begin{equation}\label{eq:ProvaCorollario}
   \mathcal{E}_n\;:=\;\Big({ 1+\frac{\nu^2}{(n+\sqrt{1-\nu^2})^2} \Big)^{\!-\frac{1}{2}}}\,.
  \end{equation}
   \item $Q_{\nu,k}(E)=0$ for 
  \[
   \begin{array}{ll}
    E=-B & \quad\textrm{if $k=-1$, and $\nu>0$} \\
    E=B  & \quad\textrm{if $k=1$ and $\nu<0$} \\
    \textrm{no value of $E$} & \quad \textrm{otherwise\,.}
   \end{array}
  \]
  \item $R_{\nu,k}(E)=0$ for 
  \[
   \begin{array}{ll}
    E=-B & \quad\textrm{if $k=1$ and $\nu<0$} \\
    E=B & \quad\textrm{if $k=-1$, and $\nu>0$} \\
    \textrm{no value of $E$} & \quad \textrm{otherwise\,.}
   \end{array}
  \]
    \item $S_\nu(E)=\infty$ for $\frac{\nu E }{\sqrt{1-E^2}}-B=-n$, $n\in\mathbb{N}_0$, and hence for $E=\mathrm{sign}(\nu)\,\mathcal{E}_{-n}$ with $\mathcal{E}_n$ defined in \eqref{eq:ProvaCorollario}.
 \end{itemize}
 
% in the l.h.s.~of \eqref{eq:roots} it is straightforward to find the following.
% \begin{itemize}
%  \item $P_\nu(E)=\infty$ for $\frac{\nu E }{\sqrt{1-E^2}}+B=-n$, $n\in\mathbb{N}_0$, and hence for $E=-\mathrm{sign}(\nu)\,\mathcal{E}_n$ with
%  \[
%   \mathcal{E}_n\;:=\;\Big({ 1+\frac{1}{n+\sqrt{1-\nu^2}} \Big)^{\!-\frac{1}{2}}}\,.
%  \]
%   \item $Q_{\nu,k}(E)=0$ for 
%  \[
%   \begin{array}{ll}
%    E=-B & \textrm{if $k=1$, irrespectively of the sign of $\nu$} \\
%    E=-B  & \textrm{if $k=-1$ and $\nu>0$} \\
%    \textrm{no value of $E$} & \textrm{if $k=-1$ and $\nu<0$\,.}
%   \end{array}
%  \]
%  \item $R_{\nu,k}(E)=0$ for 
%  \[
%   \begin{array}{ll}
%    \textrm{no value of $E$}  & \textrm{if $k=1$ and $\nu>0$} \\
%    E=-B & \textrm{if $k=1$ and $\nu<0$} \\
%    E=B & \textrm{if $k=-1$, irrespectively of the sign of $\nu$\,.}
%   \end{array}
%  \]
%    \item $S_\nu(E)=\infty$ for $\frac{\nu E }{\sqrt{1-E^2}}-B=-n$, $n\in\mathbb{N}_0$, and hence for $E=\mathrm{sign}(\nu)\,\mathcal{E}_{-n}$ with $\mathcal{E}_n$ defined as above.
% \end{itemize}
 
 Therefore, for the problem $\mathfrak{F}_{\nu,k}(E)=\infty$, which is equivalent to 
 \begin{equation*}%\label{eq:Znu}
  Z_{\nu,k}(E)\;:=\;\frac{P_\nu(E)}{S_\nu(E)} \,\frac{Q_{\nu,k}(E)}{R_{\nu,k}(E)}\;=\;\infty\,,
 \end{equation*}
  we can distinguish the following cases. 
  
For all $k$ and $\nu$, then $Z_{\nu,k}(E)=\infty$ at least for $E=-\mathrm{sign}(\nu) \mathcal{E}_n$ with $n\geqslant 1$ (which makes $P_\nu$ diverge, keeping $Q_{\nu,k}$, $R_{\nu,k}$, and $S_\nu$ finite); the remaining possibilities $E=\pm B$ have to be discussed separately. 

If $k$ and $\nu$ have the same sign, then $\lim_{E \to \pm B} Z_{\nu,k}(E)$ is either zero or infinity because only one among $P_\nu$ and $S_\nu$ diverges, $Q_{\nu,k}$ and $R_{\nu,k}$ remaining finite. Explicitly,
\[
\begin{split}
	\lim_{E \to \mp B} Z_{\nu,k}(E)\;&=\;\infty \qquad \textrm{if } \nu \gtrless 0 \\
	\lim_{E \to \pm B} Z_{\nu,k}(E)\;&=\;0 \qquad\;\: \textrm{if } \nu \gtrless 0\,.
\end{split}
\]
Thus, the value $E=-\mathrm{sgn}(\nu)B$ is admissible and $E=\mathrm{sgn}(\nu)B$ is to be discarded. This proves formula \eqref{eq:EVEnk1} for the case $k$ and $\nu$ with the same sign.

If instead $k$ and $\nu$ have opposite sign, then $\lim_{E\to \pm B} Z_{\nu,k}(E)$ must be either determined resolving the indeterminate $P_\nu\cdot Q_{\nu,k}=\infty\cdot 0$ ($R_{\nu,k}$ and $S_\nu$ being finite) or resolving the indeterminate form $S_\nu\cdot R_{\nu,k}=\infty\cdot 0$ ($P_\nu$ and $Q_{\nu,k}$ being finite). Owing to the asymptotics $\Gamma(x) \sim x^{-1}$ as $x \to 0$ all these limits are finite and non-zero, which makes the values $\pm B$ not admissible. This discussion proves formula \eqref{eq:EVEnk1} for the case in which $k$ and $\nu$ have opposite sign.
\end{proof}

% \begin{rem}
% By \cite[Eq.~13.1.3]{Abramowitz-Stegun-1964} the condition for the eigenfunction $\psi_E$ to be a polynomial in $r$ multiplied by a decreasing exponential is $a=n \in \mathbb{N}$. This condition is precisely \eqref{eq:SommerfeldFormulaSeries}. 
% \end{rem}

% \begin{figure}[h]
% \includegraphics[scale=0.35]{g1-over-g0-minus-d-2.pdf}
% \caption{Graph of the function $f(E)-d_\nu$ for $\nu=0.9$ and $\kappa=1$.}
% \end{figure}
%
%QUI GRAFICO DI GAMMA, QUELLO RITOCCATO A MANO PER ESSERE PIU CHIARO
%
%\begin{itemize}
% \item ODE approach, hypergeometryc, asymptotics
% \item general formulas for $g_1$ and $g_0$
% \item classification formula
% \item numerical plot
% \item comment on the EV's of the distinguished
% \item comment on the non-invertible extension and on the gap
% \item anti-distinguished
% \item maybe monotonicity
% \item etc. etc.
% 
%\end{itemize}
%\begin{acknowledgments}
%Acknowledgments here
%%This work was partially supported by the 2014-2017 MIUR-FIR grant ``\emph{Cond-Math: Condensed Matter and Mathematical Physics}'' code RBFR13WAET. We warmly thank ***
%\end{acknowledgments}
%
%\appendix
%
%\section{Appendix}
%\section{One more appendix}

%\bibliographystyle{aipauth} this fucking style must be commented!!!
%\bibliography{bib_ALE}

%merlin.mbs aipauth4-1.bst 2010-07-25 4.21a (PWD, AO, DPC) hacked
%Control: key (0)
%Control: author (9) reversed initials
%Control: editor formatted (0) differently from author
%Control: production of article title (0) allowed
%Control: page (1) range
%Control: year (1) truncated
%Control: production of eprint (0) enabled
\def\cprime{$'$}

\end{document}